\documentclass[a4paper,11pt]{article}
\usepackage[ansinew]{inputenc}
\usepackage{amsmath,amsfonts,amsthm,amssymb}
\usepackage{latexsym}
\usepackage{graphicx}
\usepackage{xcolor}
\usepackage{multirow}
\usepackage{subfig}
\usepackage{lineno}
\usepackage{fullpage}
\usepackage{hyperref}

\newtheorem{theorem}{Theorem}
\newtheorem{corollary}[theorem]{Corollary}

\newtheorem{proposition}[theorem]{Proposition}
\newtheorem{lemma}[theorem]{Lemma}

\newcommand{\eps}{\varepsilon}
\newcommand{\area}{\operatorname{area}}
\newcommand{\conv}{\operatorname{conv}}
\newcommand{\dist}{\operatorname{dist}}

\newcommand{\ie}{{that is}}
\newcommand{\figurenames}{{\figurename}s}

\newcommand{\new}[1]{\textcolor{red}{#1}}

\title{Rainbow polygons for colored point sets in the plane\thanks{A preliminary version was presented at the 18th Spanish Meeting on Computational Geometry (2019).}
}

\author{David Flores-Pe\~naloza\thanks{Departamento de Matem\'aticas, Facultad de Ciencias, Universidad Nacional Aut\'onoma de M\'exico, Mexico. Email: \texttt{dflorespenaloza@ciencias.unam.mx}}
\and Mikio Kano\thanks{Ibaraki University, Hitachi, Ibaraki, Japan. Email: \texttt{mikio.kano.math@vc.ibaraki.ac.jp}}
\and Leonardo Mart\'{\i}nez-Sandoval\thanks{Institut de Math\'ematiques de Jussieu - Paris Rive Gauche (UMR 7586), Sorbonne Universit\'e, Paris, France and Departamento de Matem\'aticas, Facultad de Ciencias, Universidad Nacional Aut\'onoma de M\'exico, Mexico. Email: \texttt{leomtz@im.unam.mx}}
\and David Orden\thanks{Departamento de F\'\i sica y Matem\'aticas, Universidad de Alcal\'a, Alcal\'a de Henares, Spain. Email: \texttt{david.orden@uah.es}}
\and Javier Tejel\thanks{Departamento de M\'etodos Estad\'\i sticos, IUMA, Universidad de Zaragoza, Zaragoza, Spain. Email: \texttt{jtejel@unizar.es}}
\and Csaba D. T\'oth\thanks{Department of Mathematics, California State University Northridge, Los Angeles, CA; and Department of Computer Science, Tufts University, Medford, MA, USA. Email: \texttt{csaba.toth@csun.edu}}
\and Jorge Urrutia\thanks{Instituto de Matem\'aticas, Universidad Nacional Aut\'onoma de M\'exico, Mexico. Email: \texttt{urrutia@matem.unam.mx}}
\and Birgit Vogtenhuber\thanks{Institute of Software Technology, Graz University of Technology, Graz, Austria. Email: \texttt{bvogt@ist.tugraz.at}}
}

\date{}

\begin{document}

\maketitle


\begin{abstract}
Given a colored point set in the plane, a \emph{perfect rainbow polygon} is a simple polygon that \emph{contains} exactly one point of each color, either in its interior or on its boundary. Let $\operatorname{rb-index}(S)$ denote the smallest size of a perfect rainbow polygon for a colored point set $S$, and let $\operatorname{rb-index}(k)$ be the maximum of  $\operatorname{rb-index}(S)$ over all $k$-colored point sets in general position; that is, every $k$-colored point set $S$ has a perfect rainbow polygon with at most $\operatorname{rb-index}(k)$ vertices. In this paper, we determine the values of $\operatorname{rb-index}(k)$ up to $k=7$, which is the first case where $\operatorname{rb-index}(k)\neq k$, and we prove that for $k\ge 5$,
\[
\frac{40\lfloor (k-1)/2 \rfloor -8}{19} 
\leq\operatorname{rb-index}(k)\leq 10 \bigg\lfloor\frac{k}{7}\bigg\rfloor + 11.
\]
Furthermore, for a $k$-colored set of $n$ points in the plane in general position, a perfect rainbow polygon with at most $10 \lfloor\frac{k}{7}\rfloor + 11$	
vertices can be computed in $O(n\log n)$ time.
\end{abstract}


\section{Introduction}\label{sec:intro}

Given a colored point set in the plane, in this paper we study the problem of finding a simple polygon containing exactly one point of each color. Formally, the problem we consider is the following. Let $k\ge 2$ be an integer and let $\{1,\ldots , k\}$  be $k$ distinct colors. For every $1\le i \le k$, let $S_i$ denote a finite set of points of color~$i$ in the plane. We always assume that $S_i$ is nonempty and finite for all $i\in \{1,\ldots , k\}$, and that ${S}=\bigcup_{i=1}^k S_i$ is in general position (\ie, no three points of $S$ are collinear).

For a simple polygon $P$ with $m$ vertices (or a simple $m$-gon) and a point $x$ in the plane, we say that $P$ \emph{contains} $x$ if $x$ lies in the interior or on the boundary of $P$.
Given a $k$-colored point set ${S}=\bigcup_{i=1}^k S_i$ and a simple polygon $P$ in the plane, we call $P$ a \emph{rainbow polygon} for ${S}$ if $P$ contains at most one point of each color; and $P$ will be called a \emph{perfect rainbow polygon} if it contains exactly one point of each color. The \emph{perfect rainbow polygon problem} for a colored point set $S$ is that of finding a perfect rainbow polygon with the minimum number of vertices.

One can easily check that a perfect rainbow polygon always exists for a colored point set. A way of constructing such a polygon is described below, using the following well-known property for a plane tree: From a tree $T$ embedded in the plane with straight-line edges, a simple polygon can be built by traversing the boundary of the unbounded face of $T$, placing a copy of a vertex infinitesimally close to that vertex each time it is visited, and connecting the copies according to the traversal order. One can imagine this simple polygon as the ``thickening''
of the tree. Thus, for a colored point set $S$, to build a perfect rainbow polygon we can choose one point of each color, form a star connecting one of these points to the rest, and thicken the star; see \figurename~\ref{fig:intro1}. Note that the simple polygon obtained in this way can be as close to the star as we wish, so that it contains no other points in $S$, apart from the points in $S$ that we have chosen.

However, finding a perfect rainbow polygon of minimum size for a given colored point set (where the size of a polygon is the number of its vertices) is in general much more difficult. We believe that this problem is NP-hard. Therefore, we focus on giving combinatorial bounds for the size of minimum perfect rainbow polygons. Let $\operatorname{rb-index}(S)$ denote the \emph{rainbow index} of a colored point set $S$; that is, the smallest size of a perfect rainbow polygon for $S$. We then define the \emph{rainbow index} of $k$, denoted by $\operatorname{rb-index}(k)$, to be the largest rainbow index among all the $k$-colored point sets $S$; that is,
\begin{align}
 \operatorname{rb-index}(k) = \max \ \{ \operatorname{rb-index}(S): \mbox{$S$ is a $k$-colored point set} \}. \label{def-1}
\end{align}
In other words, $\operatorname{rb-index}(k)$ is the smallest integer such that, for every $k$-colored point set $S$, there exists a perfect rainbow polygon of size at most $\operatorname{rb-index}(k)$.

The two main results in this paper are the following. First, we determine the values of the rainbow index up to $k=7$, which is the first case where $\operatorname{rb-index}(k)\neq k$, namely
\[
\begin{array}{|c|c|c|c|c|c|}
\hline
k & 3 & 4 & 5 & 6 & 7 \\
\hline
\operatorname{rb-index}(k) & 3 & 4 & 5 & 6 & 8 \\
\hline
\end{array}
\]
Second, we prove the following lower and upper bounds for the rainbow index
\[
\frac{40\lfloor (k-1)/2 \rfloor -8}{19} \leq\operatorname{rb-index}(k)\leq 10 \bigg\lfloor\frac{k}{7}\bigg\rfloor + 11. 
\]
Furthermore, for a $k$-colored set of $n$ points in the plane, a perfect rainbow polygon
of size meeting these upper bounds can be computed in $O(n\log n)$ time.

\begin{figure}[htb]
	\centering
	\subfloat[]{
		\includegraphics[scale=0.48,page=1]{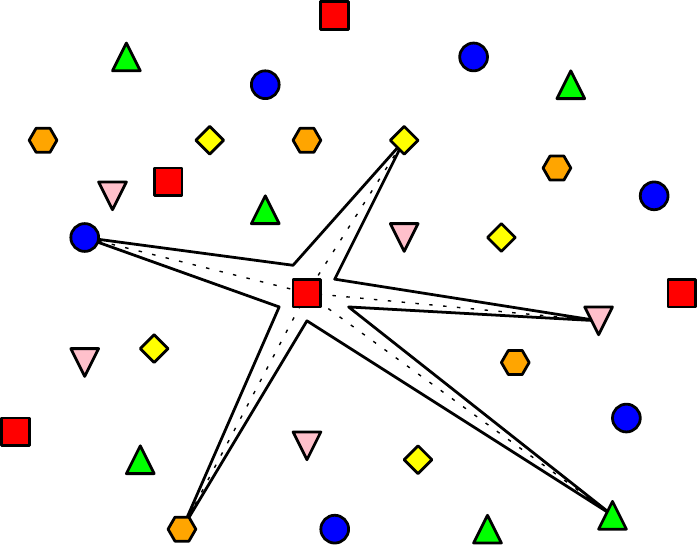}
		\label{fig:intro1}
	}~~~~~~~~~~~~~~~~~~~~
	\subfloat[]{
		\includegraphics[scale=0.52,page=2]{Figures.pdf}
		\label{fig:intro2}
	}
	\caption{(a) Thickening a tree to obtain a perfect rainbow polygon. Different colors are represented by different geometric objects. (b) A noncrossing covering tree for the eight black points that can be partitioned into five segments, $s_1=u_1u_8$, $s_2=u_2u_3$, $s_3=u_2u_4$, $s_4=u_2u_5$, and $s_5=u_6u_7$; and two forks, $u_6$ with multiplicity 1 and $u_2$ with multiplicity 2.}
	\label{fig:intro}
\end{figure}

The rainbow index for small values of $k$ is analyzed in Section~\ref{sec:small}. In Sections~\ref{sec:upper} and~\ref{sec:lower}, we provide our upper and lower bounds for the rainbow index, respectively. These bounds are based on the analysis of the complexity of noncrossing covering trees for sets of points, under a new measure defined in this paper. This measure and the relationship between perfect rainbow polygons and noncrossing covering trees is given in Section~\ref{sec:treesandpolygons}.

\subsection*{Related previous work}

Starting from the celebrated Ham-Sandwich theorem, a considerable amount of research about discrete geometry on colored point sets (or mass distributions) has been done. For instance, given $cg$ red points and $dg$ blue points in the plane, where $c$, $d$, and $g$ are positive integers, the Equitable Subdivision Theorem establishes that there exists a subdivision of the plane into $g$ convex regions such that each region contains precisely $c$ red points and $d$ blue points~\cite{BKS2000, Sakai2002}. It is also known that every $d$-colored set of points in general position in $\mathbb{R}^d$ can be partitioned into $n$ subsets with disjoint convex hulls such that the set of points and all color classes are partitioned as evenly as possible~\cite{BRSZ2019}. For a wide range of geometric partitioning results, the reader is referred to~\cite{BS2017, BHK2015, BK2012, BKS2000, BRSZ2019, HKV2017, KK2003, KK2017, Matousek, Sakai2002} and the references therein.

In addition to geometric partitions, for colored points in the plane some research focuses on geometric structures covering the points in some specific way. For instance, covering the colored points with noncrossing monochromatic matchings~\cite{DK2001}, noncrossing heterochromatic matchings~\cite{KSU2014}, noncrossing alternating cycles and paths~\cite{KPT2008}, noncrosing alternating spanning trees~\cite{BBMS2018, KSU2014} or noncrossing $K_{1,3}$ stars~\cite{Abrego2019}. In other papers, the main goal is selecting $k$ points with $k$ distinct colors (a rainbow subset) from a $k$-colored point set such that some geometric properties of the rainbow subset are maximized or minimized. Rainbow subsets with maximum diameter are investigated in~\cite{Fan2014, Ju2013}, with minimum diameter in~\cite{FX2014, Pruente2019}, and rainbow subsets optimizing matchings under several criteria are studied in~\cite{Bereg2019}. In addition, several traditional geometric problems for uncolored point sets become NP-hard for colored point sets. For instance, the following problems are NP-complete~\cite{Ju2013}: Computing a rainbow subset minimizing (maximizing) the length of its minimum spanning tree, computing a rainbow subset minimizing its convex hull, or computing a rainbow subset maximizing the distance between its closest pair.

Given a 3-colored point set $R\cup B\cup G$ consisting of red, blue, and green points in the plane, a well-known result is that there exists an empty heterochromatic triangle, where the three vertices have distinct colors~\cite{DHKS2003}. In particular, a heterochromatic triangle of minimum area cannot contain any other point from $R\cup B \cup G$ in its interior, hence its interior is empty, and its boundary contains exactly one point of each color. This implies that $\operatorname{rb-index}(3)=3$. Related work~\cite{BHK2015} deals with colored lines instead of colored points, showing that in an arrangement of 3-colored lines, there always exists a line segment intersecting exactly one line of each color. Aloupis et al.~\cite{Aloupis2011} study the problem of coloring a given point set with $k$ colors so that every axis-aligned strip containing sufficiently many points contains a point from each color class.

\section{Covering trees versus perfect rainbow polygons}\label{sec:treesandpolygons}

Given a set of (monochromatic) points, in this section we derive a lower bound for the size of simple polygons that contain the given points and have arbitrarily small area. We also provide a lower bound for the size of a perfect rainbow polygon for some colored point sets.

A noncrossing \emph{covering tree} for a set $S$ of points in the plane is a noncrossing geometric tree (\ie, a plane straight-line tree) such that every point of $S$ lies at a vertex or on an edge of the tree; see \figurename~\ref{fig:intro2}. Let $T$ be a noncrossing covering tree whose vertices can be collinear. Similarly to~\cite{DumitrescuGKT14}, we define a \emph{segment} of $T$ as a path of collinear edges in $T$. Two segments of $T$ may cross at a vertex of degree 4 or higher; we are interested in pairwise noncrossing segments. Any vertex of degree two and incident to two collinear edges can be suppressed; consequently, we may assume that $T$ has no such vertices.

Let $\mathcal{M}$ be a partition of the edges of $T$ into the minimum number of pairwise noncrossing segments.
Let $s=s(T)$ denote the number of segments in $\mathcal{M}$.
A \emph{fork} of $T$ (with respect to~$\mathcal{M}$) is a vertex $v$ that lies in the interior of a segment $ab\in \mathcal{M}$ and is an endpoint of another segment in $\mathcal{M}$. The \emph{multiplicity} of a fork $v$ is 2 if it is the endpoint of two segments that lie on opposite sides of the supporting line of $ab$; otherwise its multiplicity is 1. See \figurename~\ref{fig:intro2} for an example.

Let $t=t(T)$ denote the sum of multiplicities of all forks in $T$ with respect to $\mathcal{M}$.
We express the number of vertices in a polygon that encloses a noncrossing covering tree $T$ in terms of the parameters $s$ and $t$. If all edges of $T$ are collinear, then $s=1$ and $T$ can be enclosed in a triangle.
The following lemma addresses the case that $s\geq 2$.

\begin{lemma}\label{lem:sep}
Let $T$ be a noncrossing covering tree and $\mathcal{M}$ a partition of the edges into the minimum number of pairwise noncrossing segments.
If $s\geq 2$ and $t\geq 0$, then for every $\eps>0$, there exists a simple polygon $P$ with $2s+t$ vertices such that $\area(P)\leq \eps$ and $T$ lies in $P$.
\end{lemma}
\begin{proof}
Let $\delta>0$ be a sufficiently small constant specified below. For every vertex $v$ of $T$, let $D_v$ be a disk of radius $\delta$ centered at $v$. We may assume that $\delta>0$ is so small that the disks $D_v$, $v\in V(T)$, are pairwise disjoint, and each $D_v$ intersects only the edges of $T$ incident to $v$. Then the edges of $T$ incident to $v$ partition $D_v$ into $\deg(v)$ sectors.
If $\deg(v)\geq 3$, at most one of the sectors subtends a flat angle (\ie, an angle equal to $\pi$).
If $\deg(v)\leq 2$, none of the sectors subtends a flat angle by assumption.
Conversely, if one of the sectors subtends a flat angle, then the two incident edges are collinear;
they are part of the same segment (by the minimality of $\mathcal{M}$), and hence $v$ is a fork of multiplicity 1.

In every sector that does not subtend a flat angle, choose a point in $D_v$ on the angle bisector. By connecting these points in counterclockwise order along $T$, we obtain a simple polygon $P$ that contains $T$. Note that $P$ lies in the $\delta$-neighborhood of $T$, so $\area(P)$ is less than the area of the $\delta$-neighborhood of $T$. The $\delta$-neighborhood of a line segment of length $\ell$ has area $2\ell\delta +\pi\delta^2$. The $\delta$-neighborhood of $T$ is the union of the $\delta$-neighborhoods of its segments. Consequently, if $L$ is the sum of the lengths of all segments in $\mathcal{M}$, then the area of the $\delta$-neighborhood of $T$ is bounded from above by $2L\delta+s\pi\delta^2$, which is less than $\eps$ if $\delta>0$ is sufficiently small.

It remains to show that $P$ has $2s+t$ vertices; that is, the total number of sectors whose angle is not flat is precisely $2s+t$.
We define a perfect matching between the vertices of $P$ and the set of segment endpoints and forks (with multiplicity) in each disk $D_v$ independently for every vertex $v$ of $T$.
If $v$ is not a fork, then $D_v$ contains $\deg(v)$ vertices of $P$ and $\deg(v)$ segment endpoints.
If $v$ is a fork of multiplicity 1, then $D_v$ contains $\deg(v)-1$ vertices of $P$ and $\deg(v)-2$ segment endpoints.
Finally, if $v$ is a fork of multiplicity 2, then $D_v$ contains $\deg(v)$ vertices of~$P$ and $\deg(v)-2$ segment endpoints.
In all cases, there is a one-to-one correspondence between the vertices in $P$ lying in $D_v$ and the segment endpoints and forks (with multiplicity) in $D_v$.
Consequently, the number of vertices in $P$ equals the sum of the multiplicities of all forks plus the number of segment endpoints, which is $2s+t$, as required.
\end{proof}

Next, we establish a relation between covering trees and points sets, assuming a stronger notion of general position.
A point set is in \emph{strong general position} if
there is no nontrivial algebraic relation between the coordinates of the points.

\begin{lemma}\label{lem:tree}
Let $S$ be a finite set of points in the plane in strong general position 
and let $B$ be an axis-aligned bounding box of $S$.
Then there exists an $\eps>0$ such that \new{for every} simple polygon $P$ with $m$ vertices \new{such that $S\subset P$} and $\area(P\cap B)\leq\eps$, \new{there is} a noncrossing covering tree $T$ \new{of $S$} and a partition of the edges into pairwise noncrossing segments such that $2s+t\leq m$.
\end{lemma}
\begin{proof}
Let $m\geq 3$ be an integer such that for every $n\in \mathbb{N}$, there exists a simple polygon $P_n$ with precisely $m$ vertices such that $S\subset {\rm int}(P_n)$ and $\area(P_n\cap B)\leq \frac{1}{n}$.
Let $\overline{\mathbb{R}^2}$ be the compactification of the Euclidean plane $\mathbb{R}^2$
by adding a circle at infinity, corresponding to the unit (direction) vectors in $\mathbb{S}^1$.
By compactness, the sequence $(P_n)_{n\geq 1}$ contains a convergent subsequence of polygons in $\overline{\mathbb{R}^2}$. The limit is a \new{(possibly degenerate)}
polygon $\overline{P}$ with precisely $m$ vertices (some of which may coincide) such that $S\subset \overline{P}$ and $\area(\overline{P}\cap B)=0$; see Fig.~\ref{fig:polygons}(a).
Some vertices of $\overline{P}$ may be on the circle at infinity, and so every edge is
a single point, a line segment, a ray, or an arc of the circle at infinity.
Since $S$ is in strong general position, we may assume that if two rays are parallel, then one contains the other.
Hence any edge of $\overline{P}$ along the circle at infinity is an arc with central angle less than $\pi$.

Two edges of $\overline{P}$ \emph{overlap} if they intersect in a line segment, a ray, or a circular arc at infinity; and they \emph{cross} if they intersect in a single point that lies in the relative interior of both edges. Since $P_n$ is a simple polygon for all $n\geq 1$, the edges of $\overline{P}$ are noncrossing, although they may overlap.
Since $\area(\overline{P}\cap B)=0$, the union of the edges of $\overline{P}$ is a connected set that contains $S$.
In particular, the union of the $m$ edges of $\overline{P}$
contains a noncrossing covering tree for $S$ in $\overline{\mathbb{R}^2}$.

Let $D$ be a disk whose interior contains $S$ and all vertices of $P$ that are not at infinity. We consider the polygon $\overline{P}$ as a closed curve in $\overline{\mathbb{R}^2}$.
Let $\mathcal{P}$ be the collection of maximal \new{curves} of $\overline{P}$ in the exterior of $D$.
For any two \new{curves} in $\mathcal{P}$, the cyclic order of their endpoints along $\partial D$ cannot interleave (i.e., they cannot form an $abab$ pattern), since $P_n$ is a simple polygon for all $n\geq 1$. Each \new{curve} $\gamma \in \mathcal{P}$ contains a unique maximal \new{curve} along the circle at infinity; denote this \new{curve} by $\widehat{\gamma}$. The \new{curve} $\widehat{\gamma}$ may consist of a single vertex of $\overline{P}$, or a single edge of $\overline{P}$ (of central angle less than~$\pi$), or two or more consecutive edges of $\overline{P}$. Every \new{curve} $\widehat{\gamma}$ is homotopic, with endpoints fixed, to an arc of the circle at infinity (of central angle up to $2\pi$) between the same two vertices of $\overline{P}$; let $\mathcal{Q}$ be the set of these arcs. If two arcs in $\mathcal{Q}$ overlap, then one contains the other, otherwise the endpoints of the corresponding \new{curves} of $\overline{P}$ would interleave in $\partial D$.

\begin{figure}[htb]
	\centering
		\includegraphics[width=0.8\textwidth]{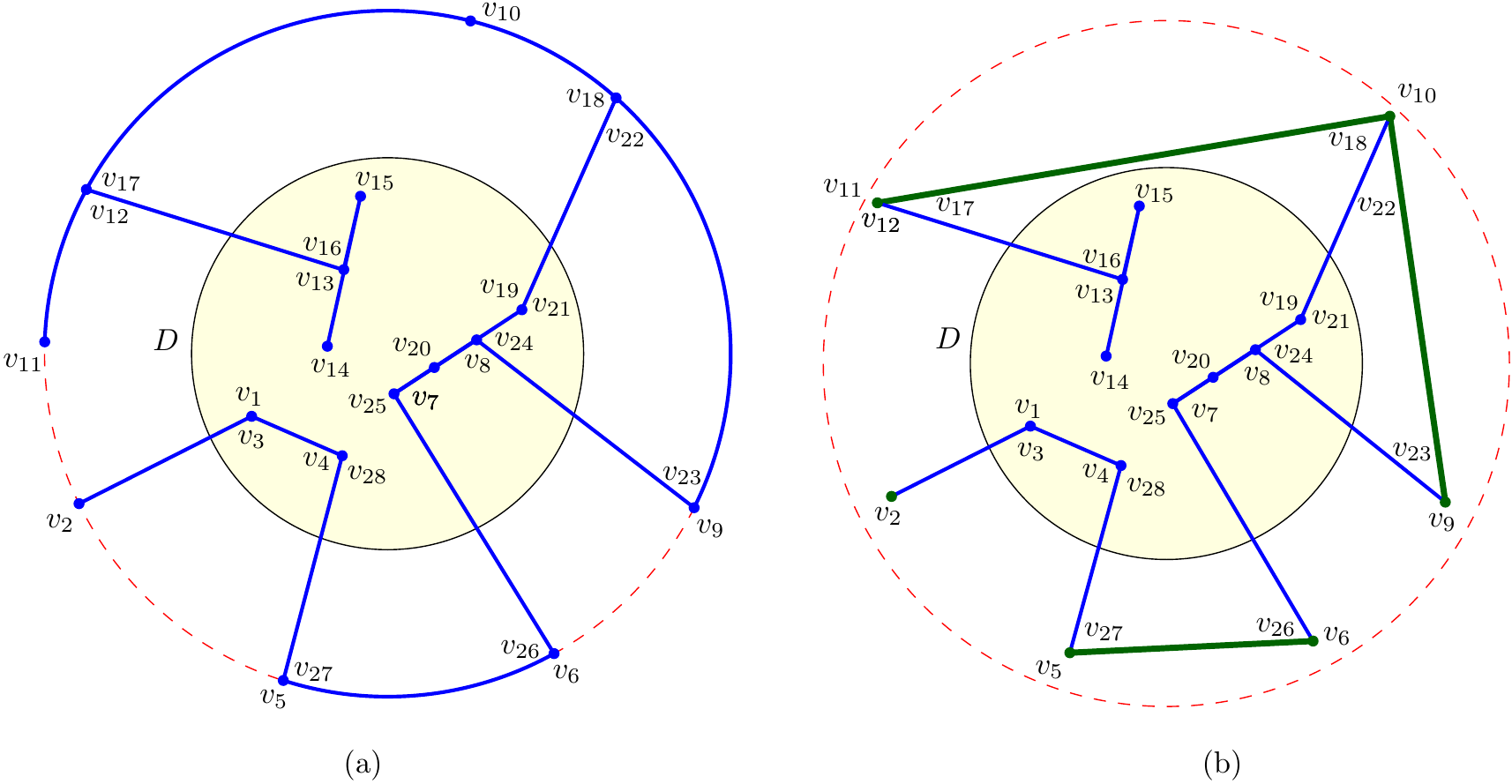}
	\caption{(a) A polygon $\overline{P}$ with 28 vertices and zero area in $\overline{\mathbb{R}^2}$. The dashed circle indicates the circle at infinity. (b) A modified polygon $P$, where all vertices are in the plane. }
	\label{fig:polygons}
\end{figure}

We modify $\overline{P}$ to obtain a polygon $P$ in the plane with $m$ vertices and pairwise noncrossing edges such that the union of its edges forms a covering tree for $S$; see Fig.~\ref{fig:polygons}(b). The transitive closure of the overlap relation between circular arcs in $\mathcal{Q}$
is an equivalence relation.
For each equivalence class of central angle less than $\pi$, we replace the arcs with a line segment in $\mathbb{R}^2$ that intersects the same set of rays. For an equivalence class of central angle at least $\pi$, we create two line segments that meet on one of the rays in the arc.

The transitive closure of the overlap relation between the edges of $P$ is an equivalence relation. The union of each equivalence class is a line segment; we call them \emph{segments} for brevity. These segments are pairwise noncrossing (since the edges of $P$ are pairwise noncrossing), and yield a covering of $S$ with a set $\mathcal{M}$ of pairwise nonoverlapping and noncrossing segments.

Analogously to the proof of Lemma~\ref{lem:sep}, at each vertex $v$ of $T$, there is a one-to-one correspondence between the vertices in $P$ located at $v$ and the segment endpoints and forks (with multiplicity) located at $v$. This implies $2s+t=m$ with respect to $\mathcal{M}$.
\end{proof}

An immediate consequence of Lemma~\ref{lem:tree} is a lower bound on the size of simple polygons with arbitrarily small area that enclose a point set $S$.

\begin{corollary}\label{cor:lowerpolygon}
Let $S$ be a finite set of points in the plane in strong general position with an axis-aligned bounding box $B$,
and let $T'$ be a noncrossing covering tree for $S$ minimizing $2s'+t' = m'$.
Then there exists an $\eps>0$ such that if $S$ is contained in a simple polygon $P$ with $m$ vertices and
$\area(P\cap B)\leq\eps$, then $m'\leq m$.
\end{corollary}
\begin{proof}
By Lemma~\ref{lem:tree}, there exists an $\eps>0$ such that if a simple polygon $P$ with $m$ vertices and $\area(P\cap B)\leq\eps$ contains $S$, then $P$ also contains a noncrossing covering tree for $S$.
Therefore, by the minimality of $T'$, necessarily $m'\le m$.
\end{proof}

A similar lower bound can be established for perfect rainbow polygons. In particular, for every set $S$ of $k$ points in the plane in strong general position
one can build a $(k+1)$-colored point set $\widehat{S}$, such that finding a noncrosing covering tree for $S$ minimizing $2s+t$ is equivalent to finding a minimum perfect rainbow polygon for $\widehat{S}$.

\begin{theorem}\label{the:lowerrainbow}
Let $S$ be a set of $k$ points in the plane in strong general position, and let $T'$ be a noncrossing covering tree for $S$ minimizing $2s'+t' = m'$. Then there exists a $(k+1)$-colored point set $\widehat{S}$ in strong general position
such that every perfect rainbow polygon for $\widehat{S}$ has at least $m'$ vertices.
\end{theorem}
\begin{proof}
Note that $m'\le 2k-2$ since a star centered at one of the points of $S$ is a covering tree for~$S$ with $k-1$ segments and no forks. Let $B$ be an axis-aligned bounding box of $S$.
We may assume, by applying a suitable affine transofmration, that $B$ is a unit square.
By Lemma~\ref{lem:tree}, there exists an $\eps>0$ such that if $S$ is contained in a simple polygon $P$ with $m$ vertices and $\area(P\cap B)\leq\eps$, then $S$ admits a noncrossing covering tree and a partition of its edges into segments such that $2s+t\leq m$.

We construct a $(k+1)$-colored point set $\widehat{S}=S\cup S_{k+1}$, where $S_{k+1}$ is a point set such that every triangle $\Delta$ with $\area (\Delta\cap B)\geq \eps/(2k)$ contains at least two points in $S_{k+1}$.
Each point of~$S$ has a unique color and all points in $S_{k+1}$ have the same color.
Let $S_{k+1}=(\frac{\eps}{16k}\cdot \mathbb{Z}^2)\cap B$, that is, a section of a integer grid in $B$.
For a triangle $\Delta$ with $\area (\Delta\cap B)\geq \eps/(2k)$, the intersection $\Delta\cap B$ is a convex polygon,
hence \new{the} boundary of $\Delta\cap B$ intersects the interor of at most $4\cdot(16k/\eps)$ grid cells,
and the area of the grid polygon $\conv(\Delta\cap S_{k+1})$ is at least $\eps/(2k)-4\eps/(16k)=\eps/(4k)$.
By Pick's theorem, at least two points of the grid $S_{k+1}$ lie in the interior of
$\conv(\Delta\cap S_{k+1})$, hence in the interior of $\Delta$.
A random perturbation of $S_{k+1}$ maintains these properties, and the
resulting $(k+1)$-colored point set $\widehat{S}$ is in strong general position.
(This construction yields $|S_{k+1}|=\Theta(k^2(\area(B))/\eps^2)$. A substantially smaller point set $S_{k+1}$
can \new{be} constructed with the same properties using classical discrepancy theory~\cite{Matousek99}.)

Now suppose, for the sake of contradiction, that there exists a perfect rainbow polygon $P$ for $\widehat{S}$ with $x$ vertices where $x< m'$. Triangulate $P$ arbitrarily into $x-2$ triangles.
Since each triangle $\Delta$ contains at most one point from $S_{k+1}$, we have $\area(\Delta\cap B)\leq \eps/(2k)$. Summation over triangles yields $\area(P\cap B)\leq (x-2)\eps/2k\leq \eps$.
By the choice of $\eps$, $S$ admits a noncrossing covering tree and a partition of its edges into segments such that $2s+t\le x$. This contradicts the minimality of $T'$, which completes the proof.
\end{proof}

We conjecture that both problems, finding a noncrossing covering tree minimizing $2s+t$ for a given point set and finding a minimum perfect rainbow polygon for a given colored point set, are NP-hard.
Many geometric variants of the classical set cover problem are known to be NP-hard. For example
covering a finite set of points by the minimum number of lines is APX-hard~\cite{Broden01,Kumar00,MT82}, see also~\cite{DM15,Kratsch16}.
The minimum-link covering problem (finding a covering path for a set of points with the smallest number of segments) is NP-complete~\cite{Arkin}. However, in these problems, the covering objects (lines or edges) may cross. There are few results on covering points with noncrossing segments. It is known, for example, that it is NP-hard to find a maximum noncrossing matching in certain geometric graph~\cite{Aloupis13}. The problem of, given an even number of points, finding a noncrossing matching that minimizes the length of the longest edge is also known to be NP-hard~\cite{abu2014bottleneck}.

\section{Rainbow indexes of k = 3, 4, 5, 6, 7}\label{sec:small}

This section is devoted to determining the rainbow indexes $\operatorname{rb-index}(k)$ up to $k=7$. The following theorem is the main result of this section, and it summarizes the results proven below.

\begin{theorem}\label{thm-1}
The rainbow indexes of $k=3,4,5,6,7$ are the following:
$\operatorname{rb-index}(3)=3$, $\operatorname{rb-index}(4)=4$, $\operatorname{rb-index}(5)=5$, $\operatorname{rb-index}(6)=6$, and $\operatorname{rb-index}(7)=8$.
Furthermore, for every $k$-colored set of $n$ points, where $3\leq k\leq 7$, a perfect rainbow polygon of size at most $\operatorname{rb-index}(k)$ can be found in $O(n\log n)$ time.
\end{theorem}

\begin{figure}[htb]
	\centering
	\subfloat[]{
		\includegraphics[scale=0.48,page=3]{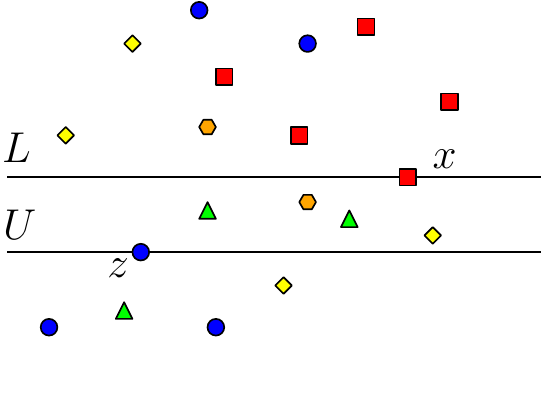}
		\label{fig:strip1}
	}~~~~
	\subfloat[]{
		\includegraphics[scale=0.48,page=4]{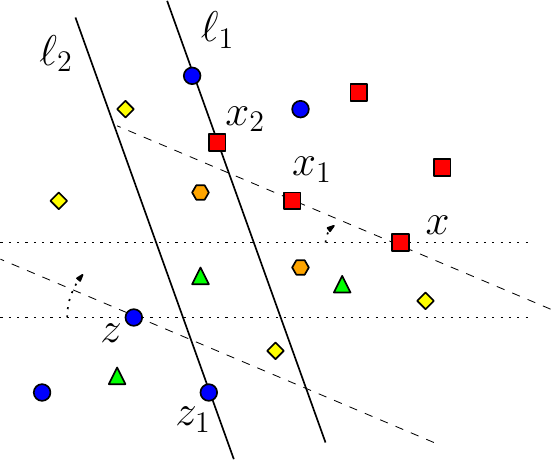}
		\label{fig:strip2}
	}~~~~
	\subfloat[]{
		\includegraphics[scale=0.48,page=5]{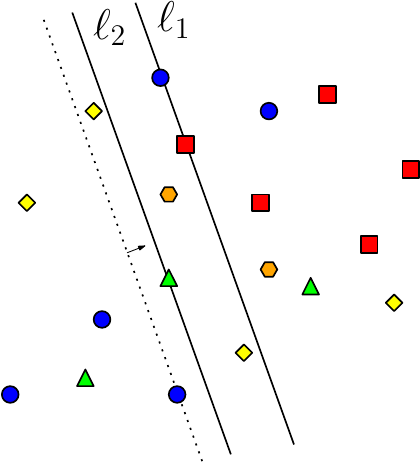}
		\label{fig:strip3}
	}
	\caption{Illustrating the proof of Lemma~\ref{lem:strip}. The points of $S_1$ and $S_3$ are drawn as red squares and blue circles, respectively.}
	\label{fig:strip}
\end{figure}

Our proof for Theorem~\ref{thm-1} relies on the following lemma (Lemma~\ref{lem:strip}), which may be of independent interest. Lemma~\ref{lem:strip} guarantees the existence of a strip containing at least one point of each color, with the additional property that there are at least two color classes that have only one point in the strip.

Before proving the lemma, we introduce some notation. The line segment connecting two points $x$ and~$y$ in the plane will be denoted by $xy$ (or $yx$). Further, a ray emanating from $x$ and passing through $y$ is denoted by $\overrightarrow{xy}$. Given two parallel lines $\ell_1$ and $\ell_2$, defining a strip $ST$, we denote by $\overline{ST}$ the closure of the strip; that is, the set of points in the interior of the strip or on the lines $\ell_1$ and $\ell_2$.

\begin{lemma} \label{lem:strip}
For every $k$-colored point set ${S}=\bigcup_{i=1}^k S_i$, where $k\ge 3$, there exist three different colors, $i_1$, $i_2$, and $i_3$, and two parallel lines $\ell_1$ and $\ell_2$  defining a strip $ST$, that satisfy the following properties:
\begin{itemize}\itemsep 0pt
  \item[i)] $\overline{ST}$ contains at least one point of $S_i$ for $i=1, \ldots , k$.
  \item[ii)] $\ell_1$ passes through a point $x\in S_{i_1}$ and a point $y\in S_{i_2}$,
    such that $x$ and $y$ are the only points with colors $i_1$ and $i_2$ in $\overline{ST}$, respectively.
  \item[iii)] $\ell_2$ passes through a point $z\in S_{i_3}$.
  \item[iv)] If $\ell_2$ passes through no other point in $S$,
   then $z$ is the only point of color $i_3$ in $\overline{ST}$.
  \item[v)] If $\ell_2$ passes through another point $w\in S$, then either $w\in S_{i_3}$ and $z,w$ are the only points of color $i_3$ in $\overline{ST}$, or $w\in S-(S_{i_1}\cup S_{i_2}\cup S_{i_3})$ and $z$ is the only point of color $i_3$ in $\overline{ST}$.
\end{itemize}
Such a strip $ST$ can be computed in $O(n\log n)$ time, where $n=|S|$.
\end{lemma}

\begin{proof}
For $i=1,\ldots k$, let $L_i$ be the horizontal line passing through the lowest point in $S_i$.
Without loss of generality, assume that $L=L_1$ is the highest line among $L_i$, $i\in \{1,\ldots, k\}$.
For $i\in \{2,\ldots, k\}$ let $S_i^-$ denote the set of points in $S_i$ that lie strictly below $L$;
and let $U_i$ be the horizontal line that passes through the highest point in $S_i^-$.
Without loss of generality, assume that $U=U_3$ is the lowest line among all $U_i$, $i\in \{2,\ldots , k\}$
	and that $L\cap S_1$ is to the right of~$U\cap S_3$; see \figurename~\ref{fig:strip1}.
Observe that by choosing $i_1=1$, $i_3=3$, $\ell_1 = L$ and $\ell_2 = U$, all conditions (i) and (iii)--(v) are satisfied.

Let $A=\mathrm{conv}(S_1)$ and $B=\mathrm{conv}(S_3^-)$.
We describe a sweepline algorithm in which we maintain a strip $\overline{ST}$ between two parallel lines $\ell_1$ and $\ell_2$.
Initially, $\ell_1=L$ and $\ell_2=U$ are horizontal lines.
We also maintain the invariants that
\begin{enumerate}\itemsep 0pt
\item[(I1)] $\ell_1$ is tangent to $A$,
\item[(I2)] $\ell_2$ is a tangent to $B$,
\item[(I3)] $|\overline{ST}\cap S_i|\geq 1$ for all $\{1,\ldots , k\}$, and
\item[(I4)] $ST\cap S_3=\emptyset$.
\end{enumerate}

During the algorithm $\ell_1$ rotates clockwise about the point in $\ell_1\cap S_1$,
and $\ell_2$ rotates clockwise about the point $\ell_2\cap S_3^-$, which are called the
\emph{pivot} points of $\ell_1$ and $\ell_2$, respectively.
Using a fully dynamic convex hull data structure~\cite{DynamicHull19}, we maintain the convex hull of the points of $S\setminus (A\cup B)$ in $ST$, above $\ell_1$, and below $\ell_2$, respectively.
By computing tangent lines from the two pivot points to the three convex hulls, we can maintain an event queue of
when the next point in $S\setminus (A\cup B)$ enters or exits the strip $ST$ (it is deleted from one convex hull and inserted into another) and when the pivot $\ell_1\cap S_1$ or $\ell_2\cap S_3^-$ must be updated.
We also maintain the number of points in $ST\cap S_i$ for $i=\{2\}\cup \{4,\ldots  ,k\}$, in order to keep track of whether invariant (I3) is satisfied.

Specifically, the rotation stops if the number of points in $ST$ of a color in $\{2\}\cup \{4,\ldots, k\}$ drops to zero,
or if the line $\ell_1$ passes through a point of color 3. Termination is guaranteed, since when $\ell_1=\ell_2$
is the common tangent of $A$ and $B$, then $ST=\emptyset$, violating invariant (I3).
Each of the lines $\ell_1$ and $\ell_2$ sweeps through every point in $S\setminus (A\cup B)$ at most once.
Since the dynamic convex hull data structures can be updated in $O(\log n)$ amortized time
and they support tangent queries in $O(\log n)$ time~\cite{DynamicHull19},
the sweepline algorithm runs in $O(n\log n)$ total time.

Note that (I1) and (I2) are always satisfied by construction.
When the sweep line algorithm terminates, we face two possible scenarios.
If (I3) is violated, then there is a color $i_2\in \{2\}\cup \{4,\ldots , k\}$ such that
$S_{i_2}$ has exactly one point in $\overline{ST}\cap S_{i_2}$ lying on $\ell_1$ or $\ell_2$.
In this case our proof is complete with condition (ii) satisfied.
Otherwise (I4) is violated; that is, $\ell_1$ passes through a point of color 3.
In this case, we can translate $\ell_2$ towards $\ell_1$ until it passes through the
last point $u$ in $ST$ of one of the colors ${i_2}\in \{2\}\cup \{4,\ldots ,k\}$; see \figurename~\ref{fig:strip3}.
Exchanging the roles of colors ${i_2}$ and ${i_3}=3$, the lemma follows.
\end{proof}

Notice that if $k=3$, then the strip defined by $\ell_1$ and $\ell_2$ in Lemma~\ref{lem:strip} is empty, so the triangle $\triangle xyz$ is empty. As a consequence, Lemma~\ref{lem:strip} provides an alternative proof for  $\operatorname{rb-index}(3)=3$.

In the remainder of this section, we refer to colors 1, 2, 3, 4, 5, 6, and 7 (if they exist) as \emph{red}, \emph{blue}, \emph{green}, \emph{yellow}, \emph{pink}, \emph{orange}, and \emph{black}, respectively (e.g., a 4-colored point set will be red, blue, green, and yellow). Furthermore, when applying Lemma~\ref{lem:strip}, we may assume without loss of generality that the colors ${i_1}$, ${i_2}$, and ${i_3}$ are red, blue, and green, respectively, the lines $\ell_1$ and $\ell_2$ are horizontal, the point $x$ is to the left of point $y$ on $\ell_1$, and if $\ell_2$ passes through another point $w$ of $S$ that is not green, then $w$ is yellow. In addition, if $p$ is the intersection point between a ray $\overrightarrow{zu}$ and a line $\ell$, then $p'$ will denote a point infinitesimally close to $p$ on the ray $\overrightarrow{zu}$ towards $z$; see \figurename~\ref{fig:rb43}.

We can now
show that $\operatorname{rb-index}(4)=4$.

\begin{proposition}
\label{prop:rb4}
$\operatorname{rb-index}(4)=4$.
\end{proposition}

\begin{proof} We first show that $\operatorname{rb-index}(4)\ge 4$. Consider the $4$-colored point set in \figurename~\ref{fig:rb41}, where $S_1=\{x\}$, $S_2=\{y\}$, $S_3=\{z\}$, and $S_4$ consists of two points in the interior of the triangle $\triangle xyz$. Every triangle that contains a point of color 1, 2, and 3 must contain $\triangle xyz$,
hence two points of $S_4$. It follows that there exists no perfect rainbow triangle.

\begin{figure}[htb]
	\centering
	\subfloat[]{
		\includegraphics[scale=0.45,page=6]{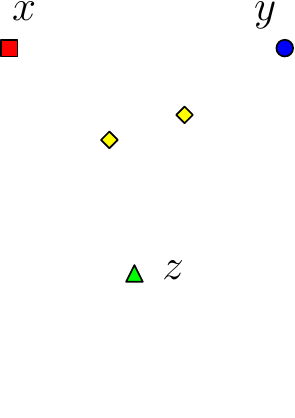}
		\label{fig:rb41}
	}~~~~
	\subfloat[]{
		\includegraphics[scale=0.45,page=7]{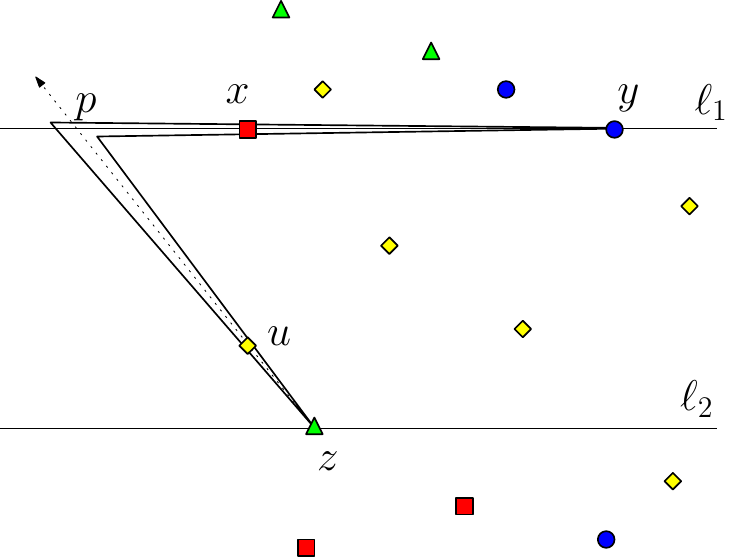}
		\label{fig:rb42}
	}~~~~
	\subfloat[]{
		\includegraphics[scale=0.45,page=8]{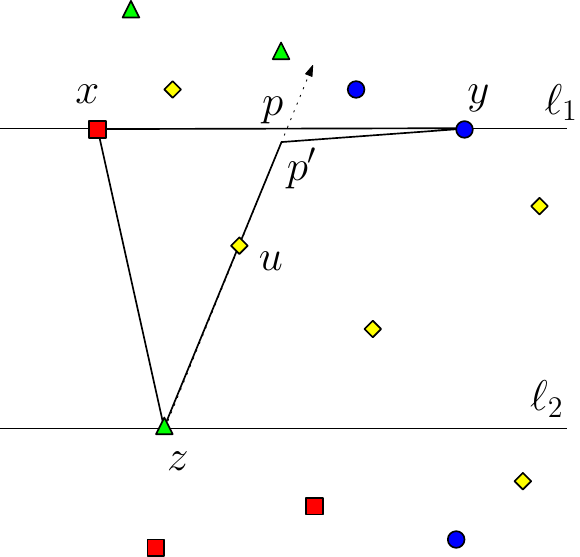}
		\label{fig:rb43}
	}
	\caption{Illustrating the proof of Proposition~\ref{prop:rb4}.}
	\label{fig:rb4}
\end{figure}

We now show that $\operatorname{rb-index}(4)\le 4$. Let $S=S_1\cup S_2 \cup S_3 \cup S_4$ be a point set in the plane whose points are colored red, blue, green, and yellow. By Lemma~\ref{lem:strip},
there is a strip defined by two horizontal lines, $\ell_1$ and $\ell_2$, where $\ell_1$ passes through a red point $x$ and a blue point $y$, and $\ell_2$ passes through a green point $z$, such that either there are only yellow points in the interior of the strip, or the strip is empty and $\ell_2$ passes through a yellow point $w$. In the first case, we rotate the horizontal ray emanating from $z$
	clockwise until it encounters a yellow point $u$ in the interior of the strip; see \figurename~\ref{fig:rb42}.
Let $p$ be the intersection point of $\overrightarrow{zu}$ and~$\ell_1$. By symmetry, we may assume that $p$ is to the left of $x$ or on the segment $xy$. If $p$ is to the left of $x$, then $py\cup pz$ is a covering tree $T$
for $\{x,y,z,u\}$,
	which can be thickened to a perfect rainbow quadrilateral by Lemma~\ref{lem:sep}; see \figurename~\ref{fig:rb42}.
If $p$ is on $xy$, then $yxzp'$ is a perfect rainbow quadrilateral; see \figurename~\ref{fig:rb43}.
Finally, if the strip is empty and $\ell_2$ contains a yellow point $w$, then $xyzw$ is a perfect rainbow quadrilateral.
\end{proof}

Before moving to the next proposition, let us prove the following useful lemma that in fact works for monochromatic
points.

\begin{lemma}
\label{lem-2}
Let $P=\{u,v,w\}$ be a set of three points in the interior of a triangle $\triangle xyz$ such that
$P\cup \{x,y,z\}$ is in general position. Then one can label the points in $P$ by $a$, $b$, and $c$ so that
$a$ lies on the line segment $xr$, for $r=\overrightarrow{xa}\cap \overrightarrow{yb}$, $b$ lies on the line segment $ys$, for $s=\overrightarrow{yb}\cap \overrightarrow{zc}$, and $c$ lies on the line segment $zt$, for $t=\overrightarrow{zc}\cap \overrightarrow{xa}$.
\end{lemma}

\begin{proof}
We first show that we can label the three points in $P$ by $a$, $b$, and $c$ so that $\triangle xya$, $\triangle yzb$, and $\triangle zxc$ are interior-disjoint. By symmetry, we may assume that the line passing through $u$ and $w$ intersects edges $xy$ and $xz$ of $\triangle xyz$, and $u$ is closer to $xy$ than $w$.
Let $p=\overrightarrow{yu}\cap \overrightarrow{zw}$. Since $P\cup \{x,y,z\}$ is in general position, $v$ lies neither on $\overrightarrow{yu}$ nor on $\overrightarrow{zw}$; see \figurename~\ref{fig:DisjointTriangles}.
Three cases arise: If $v$ is in $\triangle yzp$, then the triangles $\triangle xyu$, $\triangle yzv$, and $\triangle zxw$ are interior-disjoint; see \figurename~\ref{fig:DisjointTriangles1}.
If $v$ is in $\triangle xyp$, then $\triangle xyv$, $\triangle yzu$, and $\triangle zxw$ are interior-disjoint; see \figurename~\ref{fig:DisjointTriangles2}. Finally, if $v$ is in $\triangle zxp$, then we can take the triangles $\triangle xyu$, $\triangle yzw$, and $\triangle zxv$; see \figurename~\ref{fig:DisjointTriangles3}. By labeling $u,v,w$ as $a,b,c$ in the first case, as $b,a,c$ in the second case, and as $a,c,b$ in the third case, this part of the proof follows.

We now show that $a$ lies on the line segment $xr$, for $r=\overrightarrow{xa}\cap \overrightarrow{yb}$, $b$ lies on the line segment $ys$, for $s=\overrightarrow{yb}\cap \overrightarrow{zc}$, and $c$ lies on the line segment $zt$, for $t=\overrightarrow{zc}\cap \overrightarrow{xa}$. If $a$ were not on the line segment $xr$, then $a$ would lie below $\overrightarrow{yb}$, and so $\triangle xya$ and $\triangle yzb$ would intersect, a contradiction; see \figurename~\ref{fig:DisjointTriangles4}. The other cases, $b$ lying on $ys$ and $c$ lying on $zt$, hold by symmetry.
\end{proof}

\begin{figure}[htbp]
	\centering
	\subfloat[]{
		\includegraphics[scale=0.36,page=9]{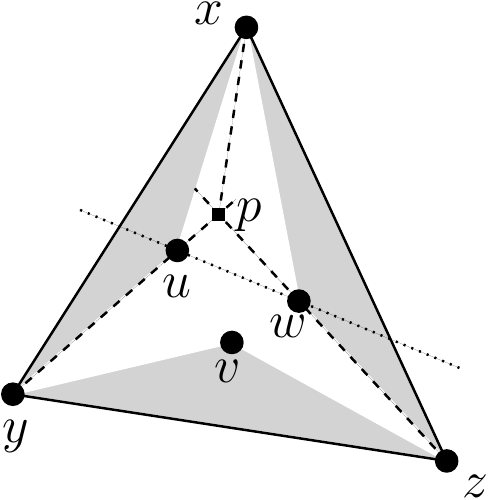}
		\label{fig:DisjointTriangles1}
	}~~~
	\subfloat[]{
		\includegraphics[scale=0.36,page=10]{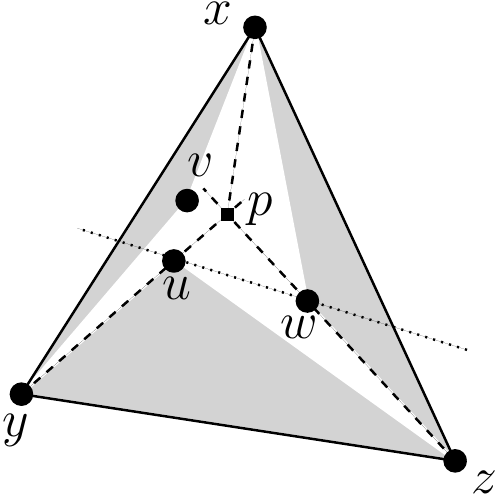}
		\label{fig:DisjointTriangles2}
	}~~~
	\subfloat[]{
		\includegraphics[scale=0.36,page=11]{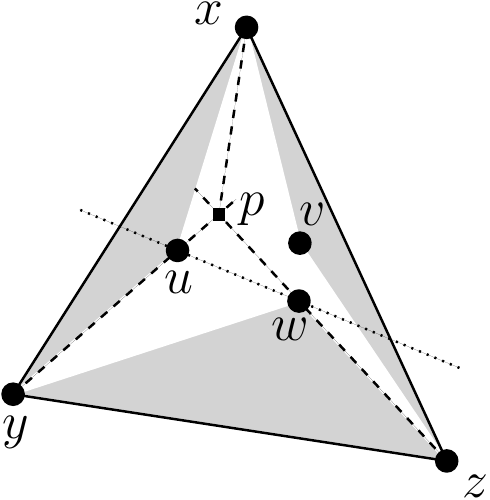}
		\label{fig:DisjointTriangles3}
	}~~~
	\subfloat[]{
		\includegraphics[scale=0.36,page=12]{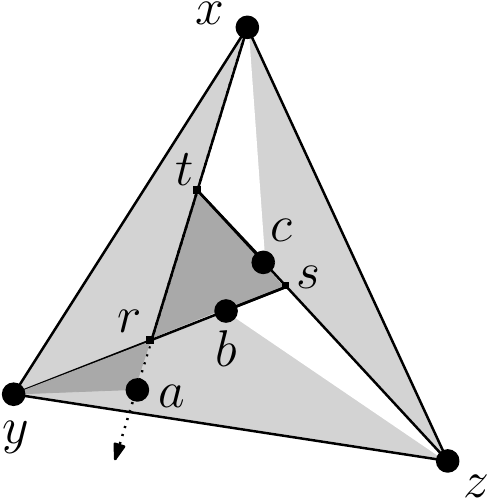}
		\label{fig:DisjointTriangles4}
	}
	\caption{Illustrating the proof of Lemma~\ref{lem-2}.}
	\label{fig:DisjointTriangles}
\end{figure}

\begin{figure}[htbp]
	\centering
	\subfloat[]{
		\includegraphics[scale=0.36,page=13]{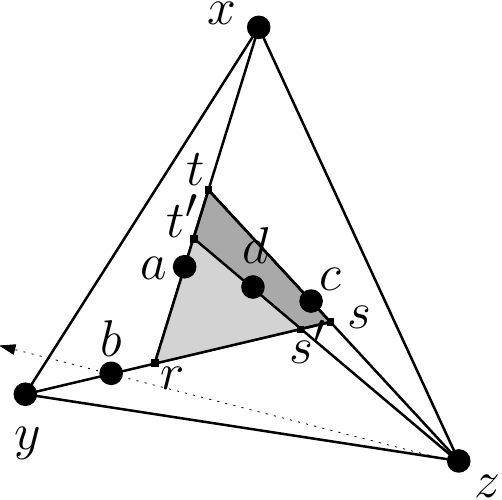}
		\label{fig:expedient1}
	}~~~
	\subfloat[]{
		\includegraphics[scale=0.36,page=14]{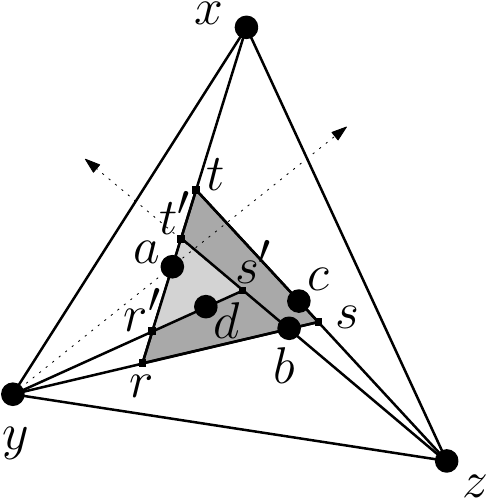}
		\label{fig:expedient2}
	}~~~
	\subfloat[]{
		\includegraphics[scale=0.36,page=15]{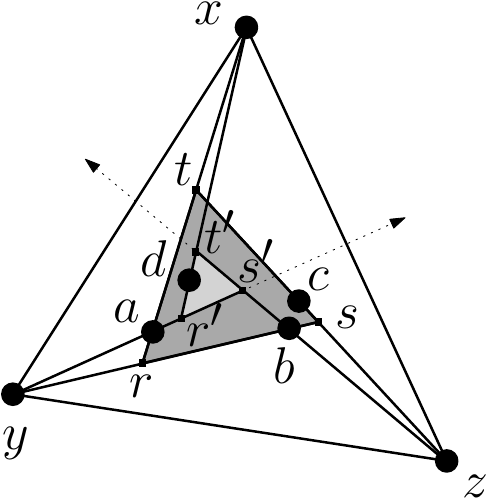}
		\label{fig:expedient3}
	}~~~
	\subfloat[]{
		\includegraphics[scale=0.48,page=16]{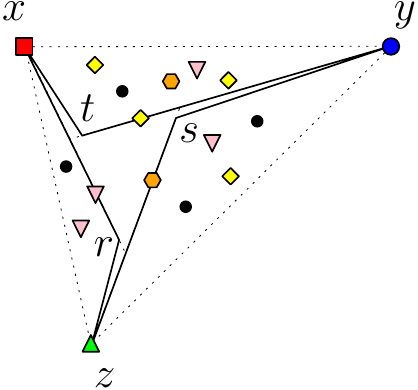}
		\label{fig:hexagon}
	}
	\caption{(a)--(c) Building a new expedient triangle $\triangle r's't'$ from the expedient triangle $\triangle rst$, when replacing $c$ by $d$. (d) A perfect rainbow hexagon for six of the colors.}
	\label{fig:speciallemma}
\end{figure}

We say that the triangle $\triangle rst$ described in Lemma~\ref{lem-2} is \emph{expedient} with respect to $\{a,b,c\}$. Note that an expedient triangle can be computed in $O(1)$ time. Further, note that the labeling given in Lemma~\ref{lem-2} is not unique, and thus expedient triangles are not uniquely determined. Using Lemma~\ref{lem-2}, one can find perfect rainbow hexagons for six colors in some special colored point sets, as the following lemma shows.

\begin{lemma}\label{lem:six}
Let $S$ be a $k$-colored set of $n$ points, with $k\ge 6$, such that $\text{conv}(S)=\triangle xyz$. Assume that
$S_1=\{x\}$, $S_2=\{y\}$, and $S_3=\{z\}$.
Then there is a perfect rainbow hexagon for six of the colors, which can be found in $O(n)$ time.
\end{lemma}

\begin{proof}
Let $a$, $b$, and $c$ be points in the interior of $\triangle xyz$ with three different colors $i_1$, $i_2$, and $i_3$, respectively, and consider an expedient triangle $\triangle rst$ with respect to $\{a,b,c\}$, which exists by Lemma~\ref{lem-2}. If there are points of $S$ in the interior of $\triangle rst$, then we choose one of them, say $d$, we replace $c$ by $d$ in $\{a,b,c\}$, and we take a new expedient triangle $\triangle r's't'$ with respect to the set $\{a,b,d\}$. Without loss of generality, we may assume that if the color of $d$ is some of $i_1$, $i_2$, and $i_3$, then $c$ is the point having the same color as $d$.

We show how to construct an expedient triangle $\triangle r's't'$ such that $\triangle r's't'\subset \triangle rst$. Assume that $a$ lies on $xr$, for $r=\overrightarrow{xa}\cap \overrightarrow{yb}$, $b$ lies on $ys$, for $s=\overrightarrow{yb}\cap \overrightarrow{zc}$, and $c$ lies on $zt$, for $t=\overrightarrow{zc}\cap \overrightarrow{xa}$; cf.\ \figurenames~\ref{fig:expedient1}--\ref{fig:expedient3}. Consider the ray $\overrightarrow{zb}$ and distinguish cases based on whether $d$ is to the left or to the right of $\overrightarrow{zb}$. If $d$ is to the right of $\overrightarrow{zb}$ (Figure~\ref{fig:expedient1}), then $\triangle rs't'$ is an expedient triangle, where $a$ lies on $xr$, $b$ lies on $ys'$, and $d$ lies on $zt'$. Suppose now that $d$ is to the left of $\overrightarrow{zb}$. Two subcases arise, depending on whether $d$ is to the left or to the right of the ray $\overrightarrow{ya}$. If $d$ is to the right of $\overrightarrow{ya}$ (Figure~\ref{fig:expedient2}), then $\triangle r's't'$ is an expedient triangle, where $a$ lies on $xr'$, $d$ lies on $ys'$, and $b$ lies on $zt'$. Finally, if $d$ is to the left of $\overrightarrow{ya}$ (Figure~\ref{fig:expedient3}), then $\triangle r's't'$ is an expedient triangle, where $d$ lies on $xr'$, $a$ lies on $ys'$, and $b$ lies on $zt'$. In all three cases, $\triangle r's't'\subset \triangle rst$, as required.

Since $\triangle r's't'\subset \triangle rst$, and since $d$ is in the interior of $\triangle rst$ but not in the interior of $\triangle r's't'$, it follows that $\triangle r's't'$ contains fewer points of $S$ than $\triangle rst$. Hence we can repeat this procedure until we find an expedient triangle that is empty of points of $S$. From this expedient triangle, we can obtain a perfect rainbow hexagon for the six colors involved, by slightly moving the vertices of the expedient triangle towards the vertices of $\triangle xyz$ as depicted in \figurename~\ref{fig:hexagon}. Furthermore, an \emph{empty} expedient triangle can be computed in $O(n)$ time: We can start with an arbitrary expedient triangle $\triangle rst$. For each point $s\in S$, we can test whether $s\in \triangle rst$ and update it to a smaller triangle $\triangle r's't'\subset \triangle rst$ if necessary in $O(1)$ time.
Consequently, a perfect rainbow polygon for six of the colors can also be found in $O(n)$ time.
\end{proof}

\begin{figure}[htbp]
	\centering
	\subfloat[]{
		\includegraphics[scale=0.52,page=17]{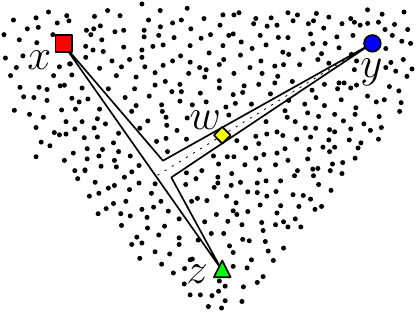}
		\label{fig:sets1}
	}~~~~~~
	\subfloat[]{
		\includegraphics[scale=0.52,page=18]{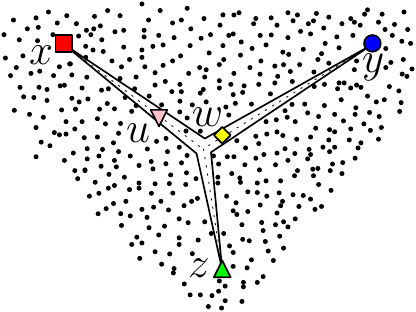}
		\label{fig:sets2}
	}~~~~~~
	\subfloat[]{
		\includegraphics[scale=0.52,page=19]{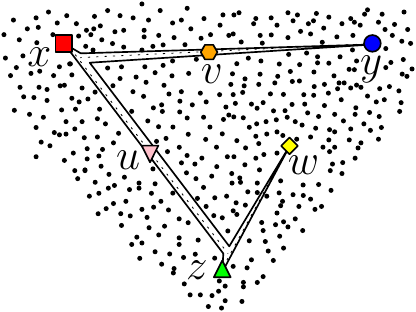}
		\label{fig:sets3}
	}
	\caption{5-, 6-, and 7-colored points sets whose rainbow indices are 5, 6, and 8, respectively.}
	\label{fig:sets}
\end{figure}

We are now ready to prove the following result.

\begin{proposition}
\label{prop:rb5}
$\operatorname{rb-index}(5)=5$, $\operatorname{rb-index}(6)=6$, and $\operatorname{rb-index}(7)=8$.
\end{proposition}
\begin{proof}
We first show that $\operatorname{rb-index}(5)\ge 5$, $\operatorname{rb-index}(6)\ge 6$, and $\operatorname{rb-index}(7)\ge 8$, by constructing a 5-colored point set $S(5)$, a 6-colored point set $S(6)$ and a 7-colored point set $S(7)$, such that $\operatorname{rb-index}(S(5))= 5$, $\operatorname{rb-index}(S(6))= 6$, and $\operatorname{rb-index}(S(7))= 8$.

The set $S(5)$ consists of four one-element color classes $S_1=\{x\}$, $S_2=\{y\}$,
$S_3=\{z\}$, and $S_4=\{w\}$, where $w$ is in the interior of $\triangle xyz$. The set $S_5$ of black points  contains $\triangle xyz$ in its convex hull, as described in the proof of Theorem~\ref{the:lowerrainbow}; that is, every triangle of area $\eps$ or more contains at least two black points; see \figurename~\ref{fig:sets1}.
The set $S(6)$ is obtained from $S(5)$ by adding a one-element color class $S_6=\{u\}$, where $u$ is in the interior of $\triangle xyz$; see \figurename~\ref{fig:sets2}.
The set $S(7)$ is based on two triangles, $\triangle xyz$ and $\triangle uvw$, where $\triangle uvw$
lies in the interior of $\triangle xyz$ and its
vertices are very close to the midpoints of the edges of $\triangle xyz$; see \figurename~\ref{fig:sets3}. The set $S(7)$ then consists of six one-element color classes, $S_1=\{x\}$, $S_2=\{y\}$, $S_3=\{z\}$, $S_4=\{u\}$, $S_5=\{v\}$, and $S_6=\{w\}$, and the dense class $S_7$ of black points as described in the proof of Theorem~\ref{the:lowerrainbow}. In the three sets, $S(5)$, $S(6)$ and $S(7)$, we assume that $x$, $y$, $z$, $u$, $v$, and $w$ (if defined) are in strong general position.

It is easy to see that a noncrossing covering tree for $\{x,y,z,w\}$ in $S(5)$, minimizing $2s+t$, requires
either two segments and a fork, or at least three segments (and no fork). Hence, by Theorem~\ref{the:lowerrainbow}, the size of a minimum perfect polygon for $S(5)$ is at least $5$. \figurename~\ref{fig:sets1} illustrates a perfect rainbow pentagon based on a covering tree that uses a segment to cover $x$ and $z$ and another segment to cover $y$ and $w$.
Every noncrossing covering tree for $\{x,y,z,w,u\}$ in $S(6)$ requires at least three segments, so the size of any perfect rainbow polygon for $S(6)$ is at least~6 by Theorem~\ref{the:lowerrainbow}. \figurename~\ref{fig:sets2} shows a perfect rainbow hexagon based on three segments that cover $x$ and $u$, $y$ and $w$, and $z$, respectively.

Finally, consider a noncrossing covering tree for $\{x,y,z,w,u,v\}$ in $S(7)$. It has at least three segments, by the pigeonhole principle, since no three points are collinear. If it has four or more segments, then the size of the corresponding perfect rainbow polygon for $S(7)$ is at least~8. Otherwise it consists of exactly three segments, and then an analysis of the possible choices shows that at least two forks are always required. Therefore, the size of a minimum perfect rainbow polygon for $S(7)$ is at least~8. \figurename~\ref{fig:sets3} illustrates the perfect rainbow octagon for $S(7)$ based on the segments that cover $\{x,u\}$, $\{y,v\}$, and $\{z,w\}$, respectively.

\begin{figure}[htb]
	\centering
	\subfloat[]{
		\includegraphics[scale=0.44,page=20]{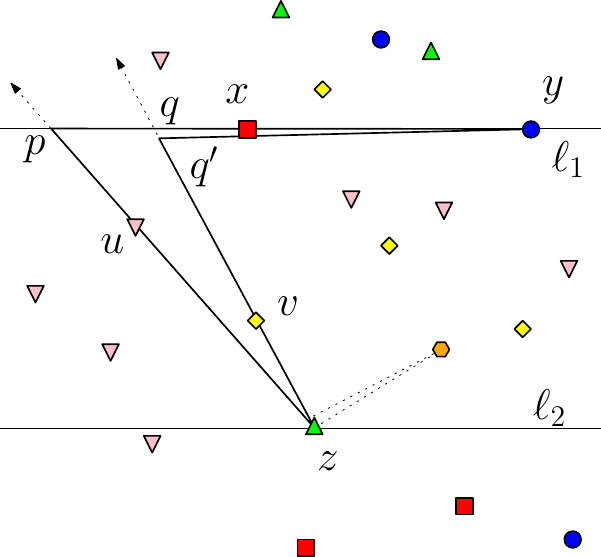}
		\label{fig:rb51}
	}~~
	\subfloat[]{
		\includegraphics[scale=0.44,page=21]{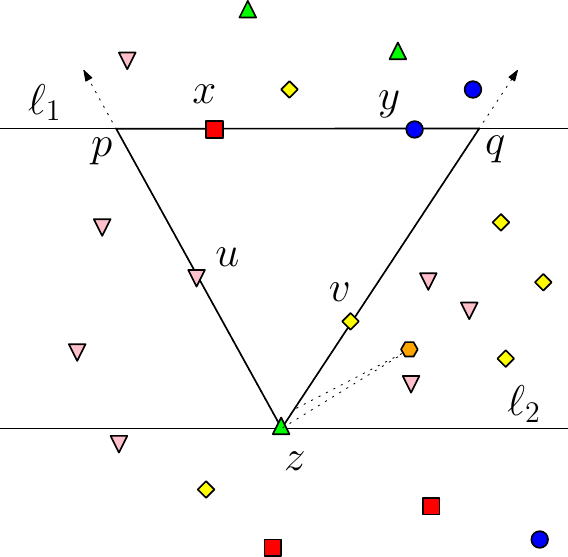}
		\label{fig:rb52}
	}~~
	\subfloat[]{
		\includegraphics[scale=0.44,page=22]{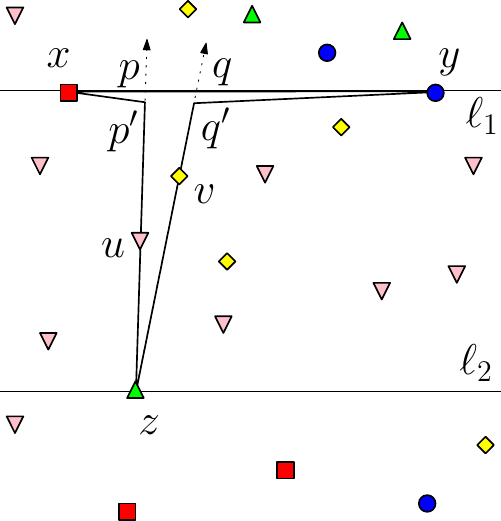}
		\label{fig:rb53}
	} \\
	\subfloat[]{
		\includegraphics[scale=0.44,page=23]{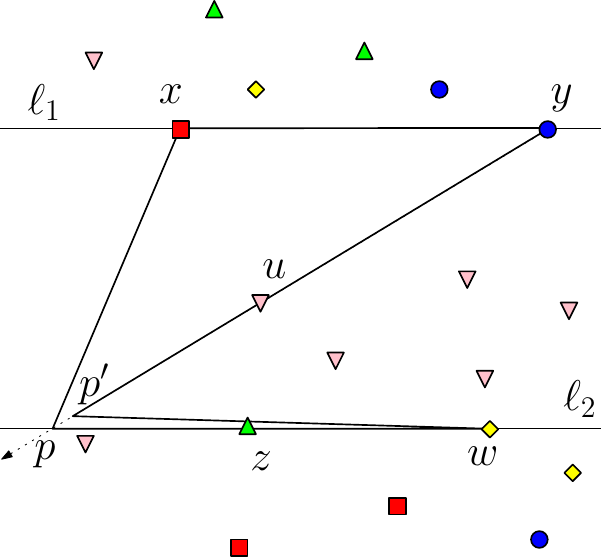}
		\label{fig:rb54}
	}~~
	\subfloat[]{
		\includegraphics[scale=0.44,page=24]{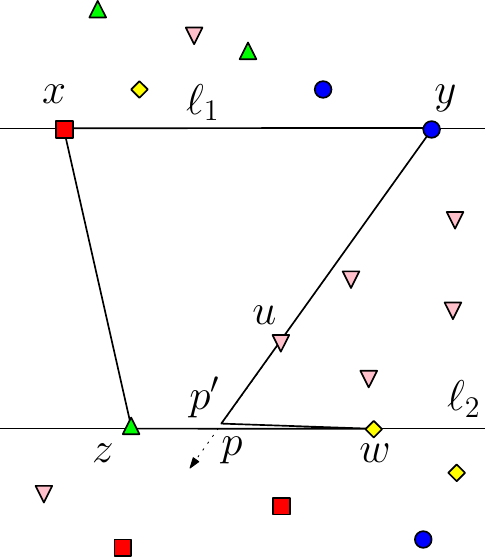}
		\label{fig:rb55}
	}~~
	\subfloat[]{
		\includegraphics[scale=0.44,page=25]{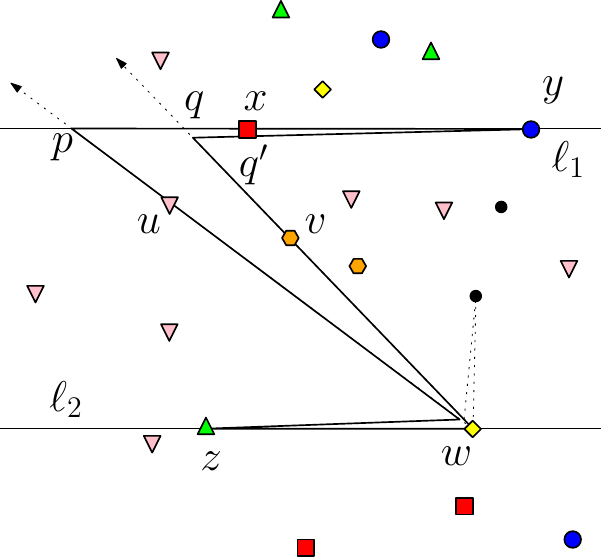}
		\label{fig:rb56}
	}
	\caption{Illustrating the proof of Proposition~\ref{prop:rb5}. }
	\label{fig:rb5}
\end{figure}

We next show that $\operatorname{rb-index}(5)\le 5$, $\operatorname{rb-index}(6)\le 6$, and $\operatorname{rb-index}(7)\le 8$. Let $S$ be a $k$-colored point set in the plane, with $k\in \{5,6,7\}$.
By Lemma~\ref{lem:strip}, there is a strip defined by two horizontal lines $\ell_1$ and $\ell_2$, with $\ell_1$ passing through a red point $x$ and a blue point $y$, and $\ell_2$ passing through a green point $z$, such that either the strip contains some points of all other colors, or $\ell_2$ passes through a yellow point $w$, the interior of the strip does not contain any other yellow point, but contains at least one point from the remaining $k-4$ color classes.
We analyse these cases in detail as follows.

\paragraph{The strip $ST$ contains points of all $k-3$ other colors.}
Consider the horizontal ray emanating from $z$ to the left and rotate it in clockwise direction, sweeping all the colored points in the strip until we find two consecutive points of $S$, say $u$ and $v$, with different colors, say yellow and pink; see \figurename~\ref{fig:rb51}. Let $p$ and $q$ be the intersection points of $\overrightarrow{zu}$ and $\overrightarrow{zv}$ with $\ell_1$, respectively. Assume that $p$ is to the left of $x$. If $q$ is also to the left of $x$ or on the line segment $xy$, then $ypzq'$ is a perfect rainbow quadrilateral for five of the colors; see \figurename~\ref{fig:rb51}. If $q$ is to the right of $y$, then $\triangle pqz$ is a perfect rainbow triangle for five of the colors; see \figurename~\ref{fig:rb52}. If $k=5$, we are done. If $k=6$, then we connect an orange point to $z$ and thicken this edge (dotted line segments in \figurenames~\ref{fig:rb51} and~\ref{fig:rb52}). In this way, we obtain a perfect rainbow hexagon in the first case and a perfect rainbow pentagon in the second case. If $k=7$, we repeat this process connecting a black point to $z$, to obtain either a perfect rainbow octagon or a perfect rainbow heptagon.

Suppose now that $p$ is to the right of $x$. Arguing in an analogous way when rotating the horizontal ray emanating from $z$ to the right counterclockwise, if $u'$ and $v'$ are two consecutive points with different colors, then we may assume that the intersection point $p_1$ between $\ell_1$ and $\overrightarrow{zu'}$ is to the left of $y$. When this happens, $p$ to the right of $x$ and $p_1$ to the left of $y$, it is straightforward to see that $\triangle xyz$ must contain at least one point of each color, and that $q$ is on $xy$. If $k=5$, then $yxp'zq'$ is a perfect rainbow pentagon; see \figurename~\ref{fig:rb53}. If $k=6$, a perfect rainbow hexagon exists by Lemma~\ref{lem:six}; see \figurename~\ref{fig:hexagon}. If $k=7$, we can build a perfect rainbow hexagon for six of the colors by Lemma~\ref{lem:six}, and form a perfect rainbow octagon for $S$ by connecting a black point to $z$ and thicken this edge.

\paragraph{The only yellow point $w$ in $\overline{ST}$ is on $\ell_2$.}
When $k=5$, the strip contains only pink points. Thus, we rotate the horizontal ray emanating from $y$ to the left
counterclockwise until it encounters a pink point $u$ in the strip; see \figurename~\ref{fig:rb54}. Let $p$ be the intersection point of $\overrightarrow{yu}$ and $\ell_2$. By symmetry, there are two cases to consider: either $p$ is to the left of $z$, or on the segment $zw$. In the first case, $yxpwp'$ is a perfect rainbow pentagon; see \figurename~\ref{fig:rb54}. In the second case $yxzwp'$ is a perfect rainbow pentagon; see \figurename~\ref{fig:rb55}.

When $k=6$ or $k=7$, the strip contains points of at least two of the colors and we can argue as before, but now rotating clockwise about $w$ instead of about $z$, to look for the first two consecutive points $u$ and $v$ with different colors, say pink and orange. If the intersection point $p$ between $\overrightarrow{wu}$ and $\ell_1$ is to the left of $x$, then we can build a perfect rainbow quadrilateral or a perfect rainbow triangle for five of the colors, as shown in \figurenames~\ref{fig:rb51} and~\ref{fig:rb52}. After that, we connect $z$ (and a black point if $k=7$) to $w$ to form a perfect rainbow hexagon (or a perfect octagon if $k=7$) for $S$; see \figurename~\ref{fig:rb56}.

Finally, if $p$ is to the right of $x$, then there are points of at least two of the colors to the left of $\overrightarrow{xw}$. Therefore, when rotating the horizontal ray emanating from $x$ to the right clockwise until finding two consecutive points $u$ and $v$ of different colors, the intersection point between $\overrightarrow{xu}$ and~$\ell_2$ will be necessarily to the right of $w$, and we can carry out symmetric constructions.
\end{proof}

The following corollary is straightforward from the proofs of Propositions~\ref{prop:rb4} and~\ref{prop:rb5}.

\begin{corollary}\label{cor:algo}
For $k=3,4,5,6,7$, a perfect rainbow polygon with at most $\operatorname{rb-index}(k)$ vertices can be found in $O(n\log n)$ time for any $k$-colored set $S$ of $n$ points.
\end{corollary}
\begin{proof}
By Lemma~\ref{lem:strip}, the strip $ST$ used in the proofs of Propositions~\ref{prop:rb4} and~\ref{prop:rb5} can be found in $O(n\log n)$ time. In addition, the cyclic order of the points in $S$ around any of $x,y,z$ or $w$ can be computed in $O(n\log n)$ time. A perfect rainbow hexagon as described in Lemma~\ref{lem:six} can be obtained in $O(n)$ time. Therefore, the corollary follows.
\end{proof}

\section{Upper bound for rainbow indexes}\label{sec:upper}

We show in this section that for every $k$-colored point set, there exists a perfect rainbow polygon of size at most $10 \lfloor\frac{k}{7}\rfloor + 11$. We begin with an auxiliary lemma showing that any seven (monochromatic) points in a vertical strip can be covered by a noncrossing forest of two trees of order four and two, respectively, such that both trees are fully contained in the strip.

\begin{lemma}\label{lem:seven}
Let $S=\{p_1, \ldots , p_7\}$ be a point set in general position in the plane, ordered by $x$-coordinate. Let $B$ be the strip defined by the two vertical lines passing through $p_1$ and $p_7$, respectively. Then, in $O(1)$ time, we can construct two noncrossing trees $T_1$ and $T_2$, of order four and two, respectively, with the following properties:
\begin{enumerate}
  \item[(i)] The union of $T_1$ and $T_2$ covers $S$ and is contained in $B$.
  \item[(ii)] For $i\in\{1,2\}$, the tree $T_i$ has a leaf $v_i$ such that the ray emanating from $v_i$ in the direction opposite to the edge incident to $v_i$ goes to the left and does not cross $T_i$. Moreover, if the extension at $v_i$ hits~$T_j$, $j\neq i$, then the extension at $v_j$ does not hit $T_i$; that is, the two trees and the two extensions do not create cycles.
\end{enumerate}
\end{lemma}

\begin{proof}
Let $\ell$ be the line passing through $p_1$ and $p_7$.
Without loss of generality, we may assume that $S$ contains at least $\lceil 5/2\rceil=3$ points below $\ell$.
Note that $p_1$ and $p_7$ are extremal points in~$S$, hence they are vertices of the convex hull $\mathrm{conv}(S)$ of $S$. Points $p_1$ and $p_7$ decompose the boundary of $\mathrm{conv}(S)$ into two convex arcs, an \emph{upper arc} and a \emph{lower arc}. Since $S$ contains at least 3 points below $\ell$, the lower arc must have at least 3 vertices (including $p_1$ and $p_7$). We distinguish between two cases depending on the number of vertices of the lower arc of $\mathrm{conv}(S)$.

\begin{figure}[htb]
	\centering
	\subfloat[]{
		\includegraphics[scale=0.60,page=26]{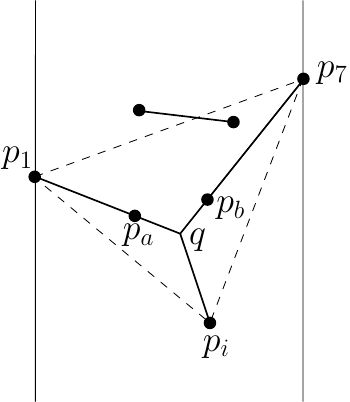}
		\label{fig:ccc1}
	}~~~~~~
	\subfloat[]{
		\includegraphics[scale=0.60,page=27]{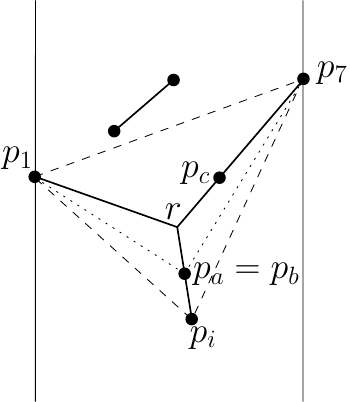}
		\label{fig:ccc2}
	}~~~~~~
	\subfloat[]{
		\includegraphics[scale=0.60,page=28]{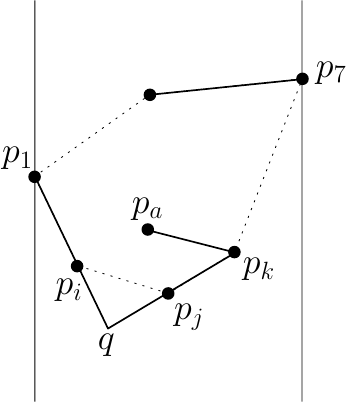}
		\label{fig:ch5}
	}
	\caption{Finding two trees when the lower arc of $\mathrm{conv}(S)$ has three vertices (\ref{fig:ccc1}--\ref{fig:ccc2}), and when it has four or more vertices (\ref{fig:ch5}).}
	\label{fig:three}
\end{figure}

\paragraph{The lower arc of $\mathrm{conv}(S)$ has 3 vertices.}
Assume that the lower arc of $\mathrm{conv}(S)$ is the path $(p_1,p_i,p_7)$, where $1<i<7$; see \figurenames~\ref{fig:ccc1}--\ref{fig:ccc2}.
Since $S$ contains at least 3 points below $\ell$, at least 2 points of $S$ are in the interior of $\triangle p_1p_ip_7$. Rotate the ray $\overrightarrow{p_1p_i}$ counterclockwise until it encounters a point $p_a$ in the interior of $\triangle p_1p_ip_7$; rotate the ray $\overrightarrow{p_7p_i}$ clockwise until it encounters a point $p_b$ in the interior of $\triangle p_1p_ip_7$. We distinguish between two cases depending on whether $p_a$ and $p_b$ are distinct.

In the first case, assume that $p_a\neq p_b$; see \figurename~\ref{fig:ccc1}. The rays $\overrightarrow{p_1p_a}$ and $\overrightarrow{p_7p_b}$ intersect in the interior of $\triangle p_1p_ip_7$, at some point $q$. By construction, the remaining two points in $S\setminus \{p_1,p_i,p_7,p_a,p_b\}$ are above the path $(p_1,q,p_7)$, in a wedge bounded by $\overrightarrow{qp_1}$ and $\overrightarrow{qp_7}$. This wedge is convex, hence it contains the line segment between the two points. Let $T_1$ be the star centered at $q$ with edges $p_1q$, $p_iq$, and $p_7q$; and let $T_2$ be the line segment spanned by the two points of $S$ above $(p_1,q,p_7)$.

In the second case, assume that $p_a=p_b$; see \figurename~\ref{fig:ccc2}. Let $r$ be the first point along the ray $\overrightarrow{p_ip_a}$ such that the line segment $p_1r$ or $p_7r$ contains a point of $S$ in the interior of $\triangle p_1p_ap_7$. Denote this point by $p_c\in S$.
By construction, the remaining two points in $S\setminus \{p_1,p_i,p_7,p_a,p_c\}$ are above the path $(p_1,r,p_7)$, in a convex wedge bounded by $\overrightarrow{rp_1}$ and $\overrightarrow{rp_7}$.
Let $T_1$ be the star centered at $r$ with edges $p_1r$, $p_ir$, and $p_7r$, and let $T_2$ be the line segment spanned by the two points of $S$ above $(p_1,r,p_7)$.

In both cases, $T_1\cup T_2$ covers all seven points in $S$, is in $B$, and is noncrossing, as required.
Moreover, the second property of the lemma is clearly satisfied by choosing $v_1=p_1$ and $v_2$ the leftmost point of $T_2$.

\paragraph{The lower arc of \boldmath{$\operatorname{conv}(S)$} has 4 or more vertices.}
Let $p_1, p_i, p_j, p_k$ be the first four vertices of the lower arc of $\mathrm{conv}(S)$ in counterclockwise order (possibly $p_k=p_7$). Let $q$ be the intersection point of the lines passing through $p_1$ and $p_i$, and through $p_j$ and $p_k$, respectively. Rotate the ray $\overrightarrow{p_kp_j}$ clockwise until it encounters a point in $S\setminus \{p_1, p_i, p_j, p_k\}$, and denote it by $p_a$, Let $T_1$ be the union of the path $(p_1,q,p_k,p_a)$, and connect the two remaining points of $S$ to define $T_2$; see \figurename~\ref{fig:ch5}.
Note that $T_1\cup T_2$ covers all seven points in $S$, is in $B$, and is noncrossing, as required. Moreover, the second property of the lemma is clearly satisfied by choosing $v_1=p_1$ and $v_2$ the leftmost point of $T_2$.
\end{proof}

Using Lemma~\ref{lem:seven}, the following theorem provides a method to find noncrossing covering trees with few segments and forks.

\begin{figure}[htb]
	\centering
	\subfloat[]{
		\includegraphics[scale=0.48,page=29]{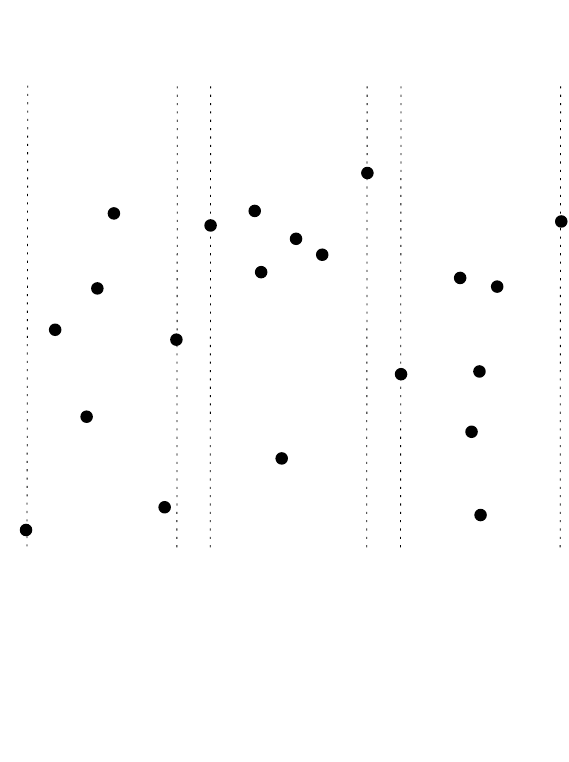}
		\label{fig:groups}
	}~~~~
	\subfloat[]{
		\includegraphics[scale=0.48,page=30]{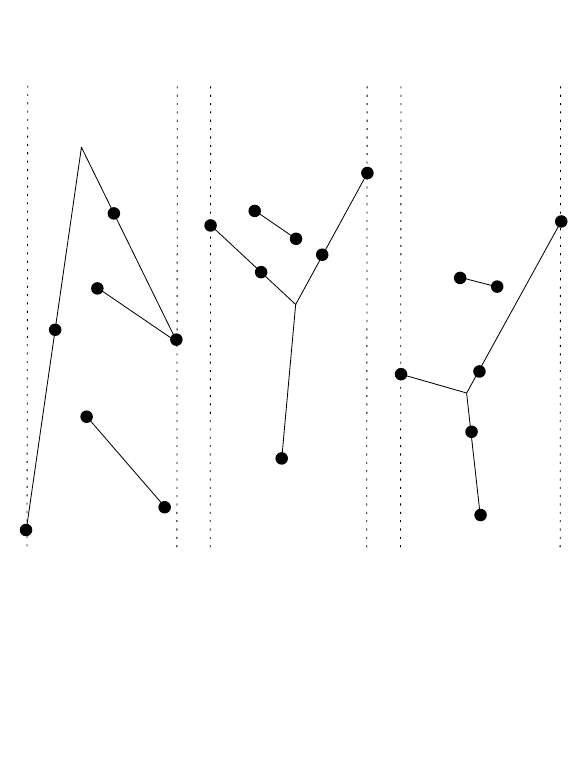}
		\label{fig:components}
	}~~~~
	\subfloat[]{
		\includegraphics[scale=0.48,page=31]{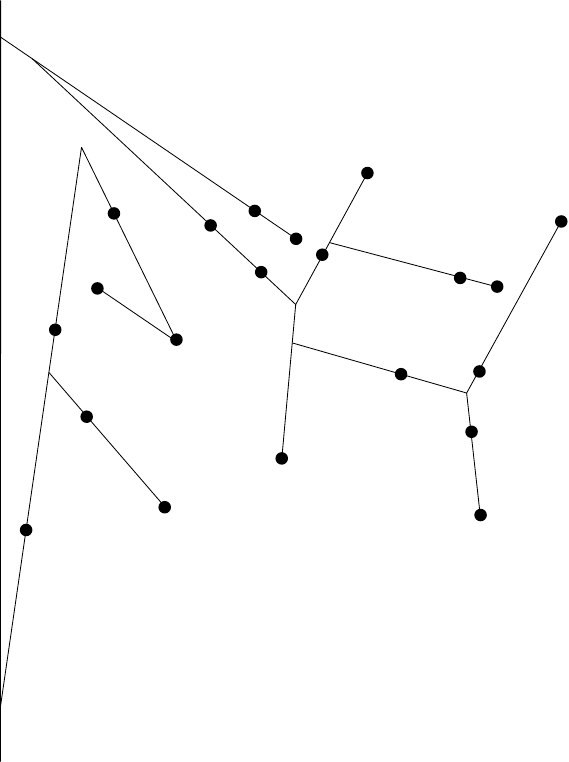}
		\label{fig:gluing}
	}

	\caption{Part (a): Dividing the $n$ points into groups of size 7. Part (b): Applying Lemma \ref{lem:seven} to each group. Part (c): Joining all trees to a vertical segment.}
	\label{fig:coveringtree}
\end{figure}

\begin{theorem}\label{the:boundtree}
Let $S$ be a finite set of $n=7j+r$ points in the plane in general position, with $j\ge 0$ and $0\le r \le 6$.
Then, in $O(n\log n)$ time, we can construct a noncrossing covering tree $T$ consisting of $4 \lfloor \frac{n}{7} \rfloor + 1 + \lceil \frac{r}{2}\rceil$ segments and $2 \lfloor \frac{n}{7} \rfloor + \lceil \frac{r}{2}\rceil$ forks with multiplicity 1.
\end{theorem}

\begin{proof}
By rotating the point set if necessary, we may assume that the points in $S$ have distinct $x$-coordinates. We assume first that $n=7j$, for some integer $j>0$. \figurename~\ref{fig:coveringtree} illustrates the method to obtain a noncrossing covering tree with $4 \lfloor \frac{n}{7} \rfloor + 1$ segments and $2 \lfloor \frac{n}{7} \rfloor$ forks with multiplicity 1. We partition the $n$ points from left to right into $j$ groups $G_1, G_2, \ldots , G_j$ of seven points each; see \figurename~\ref{fig:groups}. We apply Lemma~\ref{lem:seven} to every group $G_i$ to cover the points in $G_i$ by two trees, consisting of 4 segments in total; see \figurename~\ref{fig:components}. In this way, we obtain a forest~$F$ formed by $2j$ trees with $4j$ segments. In addition, by the same lemma, every tree $T_i$ of~$F$ contains a special leaf $v_i$ that can be extended to the left without crossing $T_i$.

We add a long vertical segment $P'$ to the left of the point set such that the extension of any tree $T_i$ of $F$ at $v_i$ crosses $P'$; see \figurename~\ref{fig:gluing}. For every tree $T_i$, we extend the edge incident to its special leaf $v_i$ to the left until the extension hits another tree, another extension, or $P'$; see \figurename~\ref{fig:gluing}. This is carried out exploring, for example, the special leaves from right to left. Thus, we join the $2j$ trees of $F$ and the segment $P'$ to form an single component. This component is necessarily a noncrossing covering tree $T$ with $4j+1$ segments and $2j$ forks with multiplicity 1, since all extensions go to the left without creating cycles. Therefore, there exists a noncrossing covering tree $T$ consisting of $4 \lfloor \frac{n}{7} \rfloor + 1$ segments and $2 \lfloor \frac{n}{7} \rfloor$ forks with multiplicity 1.

Consider the case that $n=7j+r$, where $1\le r\le 6$. Using the first $7j$ points from left to right, we proceed as before, and we build a noncrossing covering tree $T$ with $4j+1$ segments and $2j$ forks with multiplicity 1. If $j=0$, the previous step is not required. The last $r$ points can be covered by connecting the point at position $7j+1$ to the following one, the point at position $7j+3$ to the following one, and so on. If the last point cannot be connected to the following one, we assign a small horizontal segment to it. In this way, we are covering the last~$r$ points with $\lceil \frac{r}{2}\rceil$ segments. These segments can be joined to $T$ by extending their leftmost points. Therefore, we can obtain a noncrossing covering tree $T'$ consisting of $4 \lfloor \frac{n}{7} \rfloor + 1 + \lceil \frac{r}{2}\rceil$ segments and $2 \lfloor \frac{n}{7} \rfloor + \lceil \frac{r}{2}\rceil$ forks with multiplicity 1.

It remains to show that the construction above can be implemented in $O(n\log n)$ time.
We can sort the points in $S$ in increasing order by $x$-coordinates in $O(n\log n)$ time, and
hence partition $S$ into $O(n)$ groups of size seven. For each group, we can find two trees in $O(1)$ time by Lemma~\ref{lem:seven}. Finally, we can compute the left extensions of the special leaves $v_i$ of all trees $T_i$ by a standard sweepline algorithm~\cite[Sec.~2.1]{Dutch08} as follows: We sweep a vertical line~$L$ right to left from the rightmost point to one unit left of the leftmost point in $S$. In the course of the algorithm, we maintain the intersection of~$L$ with the forest~$F$, and the left extensions of all leaves $v_i$ to the right of~$L$. An event queue maintains the time steps when~$L$ passes through a vertex of $F$, when a left extension hits an edge of $F$ (in which case the left extension ends), and when two left extensions meet (in which case one extension ends and other one continues). There are $O(n)$ events, and the event queue can be updated in $O(\log n)$ time for each event. Consequently, the sweepline algorithm runs in $O(n\log n)$ time. In the last step, the sweepline~$L$ is to the left of $S$; we can let the vertical line segment $P'$ be the convex hull of the intersections of~$L$ with all surviving left extensions.
\end{proof}

Notice that by construction, the minimum number of pairwise noncrossing segments into which~$T'$ can be decomposed is precisely $4 \lfloor \frac{n}{7} \rfloor + 1 + \lceil \frac{r}{2}\rceil$. As a consequence of this theorem, we can give an upper bound for the size of a perfect rainbow polygon.

\begin{theorem}\label{the:boundpolygon}
Let $S$ be a $k$-colored set of $n$ points in general position.
Then a perfect rainbow polygon $P$ of size at most $10 \lfloor \frac{k}{7}\rfloor  + 11$
can be computed in $O(n\log n)$ time.
\end{theorem}

\begin{proof}
We choose a point of each color to define a point set $S'$ of cardinality $k=7j+r$, with $j\ge 0$ and $0\le r \le 6$. By Theorem~\ref{the:boundtree}, there is a noncrossing covering tree $T'$ for the point set $S'$, consisting of $4 \lfloor \frac{k}{7} \rfloor + 1 + \lceil \frac{r}{2}\rceil$ segments and $2 \lfloor \frac{k}{7} \rfloor + \lceil \frac{r}{2}\rceil$ forks with multiplicity 1, and it can be computed in $O(k\log k)$ time.
By Lemma~\ref{lem:sep}, given a noncrossing covering tree $T$ and a partition $\mathcal{M}$ of the edges into the minimum number $s$ of pairwise noncrossing segments, for every $\eps>0$, there exists a simple polygon $P$ with $2s+t$ vertices such that $\area(P)\leq \eps$ and $T$ lies in $P$, where $t$ is the sum of the multiplicities of all forks in $T$. Thus, for every $\eps>0$, we can construct a simple polygon $P'$ with $2 (4 \lfloor \frac{k}{7} \rfloor + 1 + \lceil \frac{r}{2}\rceil) + 2 \lfloor \frac{k}{7} \rfloor + \lceil \frac{r}{2}\rceil \le 10 \lfloor \frac{k}{7}\rfloor  + 11$ vertices such that $\area(P')\leq \eps$ and $S'$ lies in $P'$. By choosing $\eps$ sufficiently small so that $P'$ contains no other point in $S$ except for the points in $S'$, we can construct a perfect rainbow polygon for $S$ of size at most $10 \lfloor \frac{k}{7}\rfloor  + 11$.

A suitable $\eps>0$ can be half of the minimum distance between the covering tree $T'$ and the points in $S\setminus S'$. To find this distance, we can compute
the Voronoi diagram for a set of \emph{sites}, which consists of the $O(k)\leq O(n)$ edges of $T'$ and the $O(n)$ points in $S\setminus S'$ in $O(n\log n)$ time~\cite[Sec.~7.3]{Dutch08}. The Voronoi diagram is formed by $O(n)$ line segments and parabolic arcs;
and we can find the closest point in $T'$ (hence in $S\setminus S')$ for each of these arcs in $O(1)$ time.
\end{proof}

\section{Lower bound for rainbow indexes}\label{sec:lower}

For every $k\geq 3$, Dumitrescu et al.~\cite{DumitrescuGKT14} constructed a set $S$ of $n=2k$ points in the plane
in strong general position (without colors) for which
every noncrossing covering path has at least $(5n-4)/9$ edges. They also showed that every noncrossing covering tree for $S$ has at least $(9n-4)/17$ edges. Furthermore, every set of $n\geq 5$ points in general position in the plane admits a noncrossing covering tree with at most $\lceil n/2\rceil$ noncrossing segments, and this bound is the best possible. We recall that a segment is defined as a path of collinear edges.

In this section, we use the point sets constructed in~\cite{DumitrescuGKT14} to derive a lower bound for the complexity of a covering tree as defined in Section~\ref{sec:treesandpolygons}.
This bound, in turn, yields a lower bound on the complexity of perfect rainbow polygons for colored point sets built from such sets.

\begin{figure}[htbp]
			\centering
			\includegraphics[scale=0.5,page=32] {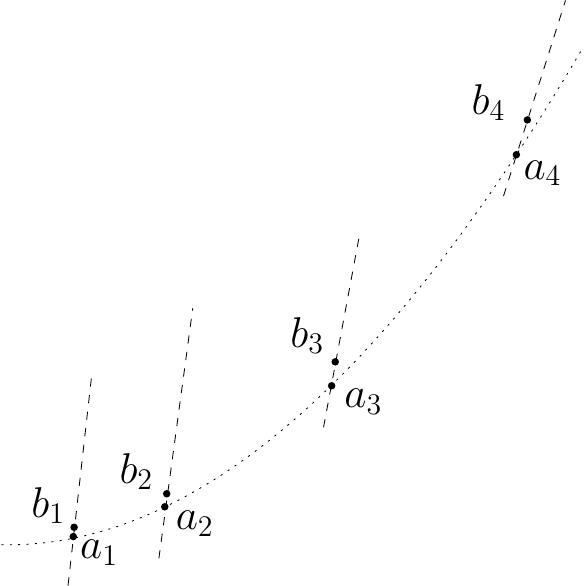}
		\caption{A sketch of the construction given in~\cite{DumitrescuGKT14} for $k=4$. The figure is not to scale.}
		\label{fig:twins}
\end{figure}

\paragraph{Construction.}
We use the point set constructed by Dumitrescu et al.~\cite{DumitrescuGKT14}.
We review some of its properties here.
For every $k\in \mathbb{N}$, they construct a set of $n=2k$ points, $S=\{a_i,b_i: i=1,\ldots , k\}$.
The pairs $\{a_i,b_i\}$ ($i=1,\ldots ,k)$ are called \emph{twins}.
The points $a_i$ ($i=1,\ldots ,k$) lie on the parabola $\alpha=\{(x,y):y= x^2\}$,
sorted by increasing $x$-coordinate. The points $b_i$ ($i=1,\ldots ,k$) lie on a
convex curve $\beta$ above $\alpha$,
such that $\dist(a_i,b_i)<\eps$ for a sufficiently small~$\eps$, and the lines $a_ib_i$ are
almost vertical with monotonically decreasing positive slopes (hence the supporting lines
of any two twins intersect below $\alpha$).
For $i=1,\ldots , k$, they also define pairwise disjoint disks $D_i(\eps)$ of radius $\eps$ centered at $a_i$ such that $b_i\in D_i(\eps)$, and the supporting lines of segments $a_ia_j$ and $b_ib_j$ meet in $D_i(\eps)$ for every $j$, $i < j \le k$.
Furthermore, (1) no three points in $S$ are collinear; (2) no two lines determined by the points in $S$ are parallel; and (3) no three lines determined by disjoint pairs of points in $S$ are concurrent.
The $x$-coordinates of $a_i$ ($i=1,\ldots ,k$) are chosen such that
(4) for any four points $c_1,c_2,c_3,c_4$ from $S$, labeled by increasing $x$-coordinate,
the supporting lines of $c_1c_4$ and $c_2c_3$ cross to the left of these points.
See \figurename~\ref{fig:twins} for a sketch of the construction.
Finally, the point set $S$ is perturbed into strong general position maintaining these
properties (except that the points $a_i$ ($i=1,\ldots ,k$) are no longer on a parabola).

\paragraph{Analysis.}
Let $S$ be a set of $n=2k$ points defined in~\cite{DumitrescuGKT14} as described above, for some $k>1$.
Let $\mathcal{M}$ be a set of pairwise noncrossing line segments in the plane
whose union is connected and contains $S$. In particular, if $T$ is a noncrossing covering tree for $S$, then any
partition of the edges of $T$ into pairwise noncrossing segments could
be taken to be $\mathcal{M}$.

A segment in $\mathcal{M}$ is called \emph{perfect} if it contains two points in $S$;
otherwise it is \emph{imperfect}.
By perturbing the endpoints of the segments in $\mathcal{M}$, if necessary,
we may assume that every point in $S$ lies in the relative interior of a segment in $\mathcal{M}$.
By the construction of $S$, no three perfect segments are concurrent, so we can define
the set $\Gamma$ of maximal paths of perfect segments; we call these \emph{perfect paths} or \emph{perfect chains}.

Dumitrescu et al.~\cite[Lemmata~4--10]{DumitrescuGKT14} proved several properties of a covering \emph{path} for $S$. Clearly, a covering path has precisely two leaves, while a covering tree may have arbitrarily many leaves. Their results are based on local configurations, however, and hold for any set of noncrossing segments $\mathcal{M}$ where the endpoints of perfect chains play the same role as the endpoints of a covering path. We restate their key results for a set $\mathcal{M}$ of noncrossing covering segments.

\begin{lemma} \label{lem:nearby} \cite[Lemma~7]{DumitrescuGKT14}
Let $pq$ be a perfect segment in $\mathcal{M}$ that contains one point from
each of the twins $\{a_i,b_i\}$ and $\{a_j,b_j\}$, where $i<j$. Assume that $p$
is the left endpoint of $pq$. Let $s$ be the segment in $\mathcal{M}$
containing the other point of the twin $\{a_i,b_i\}$. Then one of the
	following four cases occurs.
\begin{description} \itemsep -2pt
	\item[]\emph{Case 1:} $p$ is the endpoint of a perfect chain;
	\item[]\emph{Case 2:} $s$ is imperfect;
	\item[]\emph{Case 3:} $s$ is perfect, one of its endpoints $v$ lies in $D_i(\eps)$, and $v$ is the endpoint of a perfect chain;
	\item[]\emph{Case 4:} $s$ is perfect and $p$ is the common left endpoint of segments $pq$ and $s$.
\end{description}
\end{lemma}

\begin{lemma}\label{lem:vertical} \cite[Lemma~9]{DumitrescuGKT14}
Let $pq$ be a perfect segment in $\mathcal{M}$ that contains a twin
$\{a_i,b_i\}$, and let $q$ be the upper (\ie , right) endpoint of $pq$.
Then $q$ is the endpoint of a perfect chain.
\end{lemma}

Let $\Gamma_x$ be the set of maximal \emph{$x$-monotone} chains of perfect segments in $\mathcal{M}$.

\begin{lemma}\label{lem:monotone}
The right endpoints of the chains in $\Gamma_x$ are distinct.
\end{lemma}
\begin{proof}
Suppose, for the sake of contradiction, that $\gamma_1,\gamma_2\in \Gamma_x$
have a common right endpoint~$q$. Let $pq$ and $rq$, respectively,
be the rightmost segments of $\gamma_1$ and $\gamma_2$.

If $pq$ contains a twin, then $pq$ has positive slope (by construction),
and so $q$ is the upper endpoint of $pq$. In this case segment $rq$ is imperfect
by Lemma~\ref{lem:vertical}, contradicting the assumption that $rq$ is in $\gamma_2$.
We may assume that neither $pq$ not $rq$ contains a twin. In this case,
their supporting lines intersect to the left of the points in $S$ on $pq$ and $pr$
 by property (4), contradicting our assumption that $q$ is the right endpoint of both segments.
\end{proof}
\begin{corollary}\label{cor:monotone}
Every chain in $\Gamma$ consists of at most two chains in $\Gamma_x$.
\end{corollary}

Denote by $s_0$, $s_1$, and $s_2$, respectively, the number of segments in $\mathcal{M}$ that contain 0, 1, and~2 points from~$S$.
An adaptation of a charging scheme from~\cite[Lemma~4]{DumitrescuGKT14} yields the following result,
where $t$ is the number of forks (with multiplicity) in $\mathcal{M}$.

\begin{lemma}\label{lem:s01}
	$s_2\leq 8s_0+9s_1+4(t+1)$.
\end{lemma}
\begin{proof}

Let $pq$ be a perfect segment of $\mathcal{M}$, and part of a chain $\gamma\in \Gamma$.
We charge $pq$ to either an endpoint of $\gamma$ or some imperfect segment.

We define the charging as follows.
If $pq$ contains a twin, then charge $pq$ to the top vertex of
$pq$, which is the endpoint of a perfect chain by Lemma~\ref{lem:vertical}.
Assume now that $pq$ does not contain a twin, its left endpoint is $p$,
and it contains a point from each of the twins $\{a_i,b_i\}$ and $\{a_j,b_j\}$, with $i<j$.
We consider the four cases presented in Lemma~\ref{lem:nearby}.

In Case~1, charge $pq$ to $p$, which is the endpoint of a perfect chain.
In Case~2, charge $pq$ to the imperfect segment $s$ containing a point of the twin $\{a_i,b_i\}$.
In Case~3, charge $pq$ to the endpoint $v$ of a perfect chain located in $D_i(\eps)$.
Now, consider Case~4 of Lemma~\ref{lem:nearby}. In this case, $pq$ is the leftmost
segment of a maximal $x$-monotone chain $\gamma_x$.
We charge $pq$ to the right endpoint of $\gamma_x$, which is the endpoint of a perfect chain by Lemma~\ref{lem:monotone}. This completes the definition of the charges.

Note that every imperfect segment and every right endpoint of a chain in $\Gamma$
is charged at most once for perfect segments in Cases~1--3, and every left endpoint of a chain is charged at most twice.
By Corollary~\ref{cor:monotone}, each endpoint of a perfect chain is
charged at most once for perfect segments in Case~4.
Overall, every imperfect segment containing one point of $S$ is charged at most once,
and every endpoint of a perfect chain is charged at most twice.
Consequently,
\begin{equation}\label{eq:Gamma}
s_2\leq s_1+4|\Gamma|.
\end{equation}

We bound $|\Gamma|$ from above in terms of $s_0$, $s_1$, and $t$.
Choose an arbitrary root vertex in $T$, and direct all edges in $T$ towards the root.
Every perfect chain has a unique vertex $v$ closest to the root.
As all chains in $\Gamma$ are maximal and as no three perfect segments are concurrent,
$v$ must be a fork, the endpoint of an imperfect segment, or the root.
There are at most $t$ forks, the $s_0+s_1$ imperfect segments
jointly have at most $2(s_0+s_1)$ distinct endpoints, and we have one root.
This yields $|\Gamma|\leq 2(s_0+s_1)+t+1$. Combined with \eqref{eq:Gamma}, this yields,

\begin{equation*}
s_2 \leq s_1+4[2(s_0+s_1)+t+1]
=    8s_0+9s_1+4(t+1),
\end{equation*}
as claimed.
\end{proof}

\begin{lemma}\label{lem:lb}
Let $S$ be a set of $n=2k\geq 4$ points from~\cite{DumitrescuGKT14}.
Then every covering tree $T$ of $S$ satisfies $2s+t\geq (20n-8)/19$.
\end{lemma}
\begin{proof}
The combination of Lemma~\ref{lem:s01} and $n=s_1+2s_2$ yields
\begin{eqnarray*}
2s+t&=&2(s_0+s_1+s_2)+t\\
&\geq & (s_0+s_1+t)+(s_1+2s_2)\\
& = & \frac{2[8s_0+9s_1+4(t+1)] + 3s_0 + s_1 + 11t - 8}{19} + n\\
&\geq & \frac{2s_2 + s_1 - 8}{19}+ n\\
& = & \frac{n-8}{19} +n\\
& = & \frac{20n-8}{19},
\end{eqnarray*}
as claimed.
\end{proof}

We are now ready to prove the main result of this section.

\begin{theorem}\label{thm:2019}
For every integer $k\geq 5$, there exists a finite set of $k$-colored points in the plane such that every perfect rainbow polygon has at least $\frac{40\lfloor (k-1)/2 \rfloor -8}{19}$ vertices.
\end{theorem}
\begin{proof}
	Assume first that $k$ is odd, and let $S$ be the set of $k-1=2j\geq 4$ points from~\cite{DumitrescuGKT14}. If $T$ is a noncrossing covering tree for $S$ minimizing $2s+t = m$, then by Theorem~\ref{the:lowerrainbow}, there exists a $k$-colored point set $\widehat{S}$ built from $S$ such that every perfect rainbow polygon for $\widehat{S}$ has at least $m$ vertices. By Lemma~\ref{lem:lb}, every noncrossing covering tree of $S$ satisfies $2s+t \ge \frac{20(k-1)-8}{19}$, hence every perfect rainbow polygon for $\widehat{S}$ has at least $\frac{20(k-1)-8}{19} = \frac{40\lfloor (k-1)/2 \rfloor -8}{19}$ vertices.

Assume now that $k$ is even. From the $(k-1)$-colored point set $\widehat{S}$ built previously, we can obtain a $k$-colored point set $\widehat{S}'$ by adding a new point with a different color. Since every perfect rainbow polygon for $\widehat{S}$ has at least $\frac{20(k-2)-8}{19} = \frac{40\lfloor (k-1)/2 \rfloor -8}{19}$ vertices, then every perfect rainbow polygon for $\widehat{S}'$ also has at least $\frac{40\lfloor (k-1)/2 \rfloor -8}{19}$ vertices.
\end{proof}

\section{Conclusions}

In this paper, we studied the perfect rainbow polygon problem and we proved that the rainbow index of $k$ satisfies $\frac{40\lfloor (k-1)/2 \rfloor -8}{19}\leq\operatorname{rb-index}(k)\leq 10 \lfloor\frac{k}{7}\rfloor + 11$, for $k\ge 5$. We also showed that $k=7$ is the first value such that  $\operatorname{rb-index}(k) \ne k$. Our bounds are based on the equivalence between perfect rainbow polygons and noncrossing covering trees.

Several open questions arise in relation to this problem. For instance, we conjecture that given a colored point set $S$, finding a minimum perfect rainbow polygon for $S$ is NP-hard. Another interesting question is to close the gap between the lower and upper bounds on the rainbow index.

\section*{Acknowledgments}

David Flores, David Orden, Javier Tejel, Jorge Urrutia, and Birgit Vogtenhuber are supported by the H2020-MSCA-RISE project 734922-CONNECT.
Research by David Flores-Pe\~naloza was supported by the grant UNAM PAPIIT IN117317.
Research by Mikio Kano was supported by JSPS KAKENHI Grant Number 16K05248.
Research by Leonardo Mart\'\i nez-Sandoval was supported by the grant ANR-17-CE40-0018 of the French National Research Agency ANR (project CAPPS).
Research by David Orden was supported by project MTM2017-83750-P of the Spanish Ministry of Science (AEI/FEDER, UE) and by Project PID2019-104129GB-I00 / AEI / 10.13039/501100011033 of the Spanish Ministry of Science and Innovation. Research by Javier Tejel was supported by MINECO project MTM2015-63791-R, Gobierno de Arag\'on under Grant E41-17R (FEDER), and Project PID2019-104129GB-I00 / AEI / 10.13039/501100011033 of the Spanish Ministry of Science and Innovation.
Research by Csaba D. T\'oth was supported by NSF awards CCF-1422311, CCF-1423615, and DMS-1800734.
Research by Jorge Urrutia was supported by UNAM project PAPIIT IN102117.
Research by Birgit Vogtenhuber was supported by the Austrian Science Fund within the collaborative DACH project \emph{Arrangements and Drawings} as FWF project I 3340-N35.

\bibliographystyle{plainurl}
\bibliography{rainbow}

\end{document}